\documentclass[onecolumn]{IEEEtran}
\title{Local Decode and Update \\ for Big Data Compression}
\author{Shashank Vatedka,~\IEEEmembership{Member,~IEEE,}
	Aslan Tchamkerten,~\IEEEmembership{Senior~Member,~IEEE,}
	\thanks{This work was supported by Nokia Bell Labs France within the framework ``Computation over Encoded Data with Applications to Large Scale Storage.'' This work was presented in part at the 2019 IEEE International Symposium on Information Theory, Paris, France~\cite{vatedka2019local}.}
	\thanks{S.~Vatedka and A.~Tchamkerten are with the Department of Communications and Electronics, Telecom Paris, Paris, France. Email: \texttt{\{shashank.vatedka, aslan.tchamkerten\}@telecom-paristech.fr} }%
}

\usepackage{basicreq}
\usepackage[utf8]{inputenc} \usepackage[T1]{fontenc}

\begin{document}
	\maketitle

	\begin{abstract}
		This paper investigates data compression that simultaneously allows local decoding and local update. The main result is a universal compression scheme for memoryless sources with the following features. The rate can be made arbitrarily close to the entropy of the underlying source, contiguous fragments of the source can be recovered or updated by probing or modifying a number of codeword bits that is on average linear in the size of the fragment, and the overall encoding and decoding complexity is quasilinear in the blocklength of the source. In particular, the local decoding or update of a single message symbol can be performed by probing or modifying a constant number of codeword bits. This latter part improves over previous best known results for which local decodability or update efficiency grows logarithmically with blocklength. 
	\end{abstract}

	\section{Introduction}
	Recent articles~\cite{pavlichin2018quest,chen2014data,hashem2015rise}  point to the mismatch between the amount of generated data, notably genomic data~\cite{ball2012public,uk10k2015uk10k,gaziano2016million}, and hardware and software solutions for cloud storage. 
    There is a growing need for space-optimal cloud storage solutions that allow efficient remote interaction, as frequent remote access and manipulation of a large dataset can generate a large volume of internet traffic~\cite{schadt2010computational,marx2013biobigdata,stephens2015big}.
	
	Consider for instance compressing and storing DNA sequences in the cloud. If compression is handled via traditional methods, such as Lempel-Ziv~\cite{ziv1977universal,ziv1978compression}, then to retrieve say a particular gene, typically a few tens of thousands of base pairs, we need to decompress the entire DNA sequence, about three billion base pairs. Similarly, the update of a small fraction of the DNA sequence requires to update the compressed data entirely.
	Solutions have been proposed, typically using modifications of Lempel-Ziv and variants, to address some of these issues (see \emph{e.g.}, \cite{brandon2009data,deorowicz2011robust,cox2012large,deorowicz2013genome,tatwawadi2016gtrac} and the references therein).
	
	In this paper we investigate lossless data compression with both local decoding and local update properties. Accordingly, consider a rate $R$ compression of an i.i.d.$\sim p_X$ sequence $ X^n $. Let $ \ravg(s) $ denote the {\emph{average}} (over the randomness in the source $ X^n $) number of bits of the codeword sequence that need to be probed, possibly adaptively, to decode an arbitrary length $ s $ contiguous substring of $X^n$. Similarly, let $ \tavg(s) $ denote the {\emph{average}} number of codeword bits that need to be read and written, possibly adaptively, in order to update an arbitrary length $s$ contiguous substring of $X^n$. The basic question addressed here is whether it is possible to design a compression scheme such that the operations of local decoding and local update involve a number of bits that is proportional to the number of bits to be retrieved or updated. Specifically, is it possible to design a coding scheme such that, for any $R$ larger than the entropy $H(p_X)$,
		$$(\ravg(s),\tavg(s))=(O(s),O(s))\quad \text{for any $1\leq s\leq n$ ?}$$
	As we show in this paper, the answer is positive. Given $\varepsilon>0$, we exhibit a compressor, a local decoder and a local updater with the following properties: \begin{itemize}
		\item The compressor achieves rate $R=H(p_X)+\varepsilon$ universally.
		\item The local decoder achieves constant decodability $$\ravg(1)= \alpha_1\left(\frac{1}{\varepsilon^2}\log\frac{1}{\varepsilon}\right) $$ for some constant $ \alpha_1<\infty $ that only depends on $ p_X $.
		\item  the local updater achieves constant update $$ \tavg(1) = \alpha_2\left(\frac{1}{\varepsilon^2}\log\frac{1}{\varepsilon}\right) $$ for some constant $ \alpha_2 $ that only depends on $ p_X $.
		\item  For all $ s\geq 3 $ $$ \ravg(s)<s \cdot \ravg(1) $$ and $$ \tavg(s)<s \cdot \tavg(1). $$
		Moreover, if the source is non-dyadic then there exists $ \alpha_3>0 $ independent of $ n,\varepsilon $ such that for all $ s>\alpha_3/\varepsilon^2 $,  we have 
		$$ \ravg(s)< s\cdot\ravg^*(1)$$ 
		where $ \ravg^*(1) $ denotes the minimum average local decodability that can possibly be achieved by any compression scheme having rate $ R \leq H(p_X)+\varepsilon $.\footnote{We guarantee that local decompression of \emph{contiguous} substrings of the message can be performed more efficiently than repeated local decompression of the individual bits. If we want to recover $s$ arbitrary non-contiguous message symbols, it is not clear if we can simultaneously achieve rate close to entropy and $\ravg(s)<s\ravg^*(1)$.} 
		\item 
		The compression scheme has an overall encoding and decoding computational complexity that is quasilinear in $n$. 
	\end{itemize}
	We also show, through a second scheme, that it is possible to achieve $(O(\log\log n), O(\log\log n)) $ {\emph{worst-case}} local decodability and average update efficiency for any $R$ larger than the entropy $H(p_X)$ of the underlying source.
	
	
	\subsection*{Related works: word-RAM and bitprobe models}
	
	There has been a lot of work related to local decoding of compressed data structures;  see, {\it{e.g.}},~\cite{patrascu2008succincter,dodis2010changing,munro2015compressed,raman2003succinct,chandar2009locally,chandar_thesis} and the references therein. Most of these results hold under the word-RAM model which assumes that operations are on blocks of $ \Theta(\log n) $ bits, where $ n $ denotes the length of the source sequence. It is assumed that operations (memory access, arithmetic operations) on words of $ \Theta(\log n) $ bits take constant time, and the efficiency of a scheme is measured in terms of the time complexity required to perform local decoding. Therefore, algorithms in all these papers must probe $ \Omega(\log n) $ bits of the codeword even if only to recover a single bit of the source sequence.

	In the word-RAM model it is possible to compress any sequence to its empirical entropy and still be able to locally decode any message symbol in constant time~\cite{patrascu2008succincter,dodis2010changing}. In particular,  \cite{patrascu2008succincter} gives a multilevel encoding procedure that is conceptually related to our first scheme---the difference will be discussed later in Section~\ref{sec:patrascu-connections}.
	Another compression scheme is due to
	Dutta \emph{et al.}~\cite{dutta2013simple} which achieves compression lengths within a $ (1+\varepsilon) $ multiplicative factor of that of LZ78 while allowing local decoding of individual symbols in $ O(\log n+1/\varepsilon^2) $ time on average. Bille \emph{et al.}~\cite{bille2011random} gave a scheme that allows one to modify any grammar-based compressor (such as Lempel-Ziv) to provide efficient local decodability under the word-RAM model. Viola \emph{et al.}~\cite{viola2019howtostore} recently gave a scheme that achieves near-optimal compression lengths for storing correlated data while being able to locally decode any data symbol in constant time. There is a long line of work, {\it{e.g.}},~\cite{sadakane2006squeezing,gonzalez2006statistical,ferragina2007simple,kreft2010lz77}, on compression schemes that allow efficient local recovery of length $ m>1 $ substrings of the message.

	Concerning local update, Makinen and Navarro~\cite{makinen2006dynamic} designed an entropy-achieving compression scheme that supports insertion and deletion of symbols in $ O(\log n) $ time. Successive works~\cite{jansson2012cram,grossi2013dynamic,navarro2014optimal} gave improved compressors that support local decoding, updates, insertion and deletion of individual symbols in $ O(\log n/\log\log n) $ time.

	While the word-RAM model is natural for on-chip type of applications where data transfer occurs through a memory bus (generally $ \Omega(\log n) $ bits), it is perhaps less relevant for (off-chip) communication applications such as between a server, hosting the compressed data set, and the client. In this context it is desirable to minimize the number of bits exchanged, and a more relevant model is the so-called bitprobe model~\cite{nicholson2013survey} where the complexity of updating or decoding is measured by the number of bits of the compressed sequence that need to be read or modified to recover or update a single bit of the raw data. 
	
	Under the bitprobe  model, it is known that a single bit of an $ n $-length source sequence can be recovered by accessing a constant (in $ n $) number of bits of the codeword sequence~\cite{buhrman2002bitvectors,lewenstein2014improved,garg2015set,garg2017set_nonadaptive}.
	However, these works typically assume that the source sequence is deterministic and  chosen from a set of allowed sequences, and the complexity of local decoding or update is measured for the worst-case allowed sequence.
	
	The problem of locally decodable source coding of random sequences has received attention very recently. 
	Makhdoumi \emph{et al.}~\cite{makhdoumi2013locally-arxiv,makhdoumi_onlocallydecsource} showed
	that any compressor with $\rwc(1)=2$ cannot achieve a rate below the trivial rate $\log|\cX|$. Moreover, any \emph{linear} source code that achieves $ \ravg(1)=\Theta(1) $ necessarily operates at a trivial compression rate ($R=1$ for binary sources).
	Mazumdar \emph{et al.}~\cite{mazumdar2015local} gave a fixed-blocklength entropy-achieving compression scheme that permits local decoding of a single bit efficiently.  For a target rate of $ H(p_X)+\varepsilon $ the decoding of a single bit requires to probe $\ravg(1)= \Theta(\frac{1}{\varepsilon}\log\frac{1}{\varepsilon}) $ bits on the compressed codeword. They also provided a converse result for non-dyadic sources: $ \ravg(1)=\Omega(\log(1/\varepsilon)) $ for any compression scheme that achieves rate $ H(p_X)+\varepsilon $. Tatwawadi \emph{et al.}~\cite{tatwawadi18isit_universalRA} extended the achievability result to Markov sources and provided a universal scheme that achieves $\ravg(1)= \Theta(\frac{1}{\varepsilon^2}\log\frac{1}{\varepsilon}) $. It should perhaps be stressed that the papers~\cite{mazumdar2015local,tatwawadi18isit_universalRA} only investigate local decoding of a single bit and, in particular, they leave open the question whether  we can achieve  $ \ravg(s)< s\ravg^*(1) $ for $ s>1 $. It should also be noted that the construction in these papers make use of the bitvector compressor of Buhrman \emph{et al.}~\cite{buhrman2002bitvectors} which in turn is a nonexplicit construction based on expander graphs. It is also unclear whether their encoding and decoding procedures  can be peformed with low (polynomial-time) computational complexity.
	
	All the above papers on the bit-probe model consider fixed-length block coding. Variable-length source coding was investigated by Pananjady and Courtade~\cite{pananjady2018effect} who gave upper and lower bounds on the achievable rate for the compression of sparse sequences under local decodability constraints. 

	Update efficiency was studied in~\cite{montanari2008smooth}, which used sparse-graph codes to design an entropy-achieving compression scheme for which an update to any single message bit can be performed by modifying at most $ \tavg(1)=\Theta(1) $ codeword bits. The authors remarked that their scheme cannot simultaneously achieve $ \ravg(1)=\Theta(1) $ and $ \tavg(1) =\Theta(1)$. Related to update efficiency is the notion of malleability~\cite{varshney2016palimpsests,varshney2016malleable}, defined as the average fraction of codeword bits that need to be modified when the message is updated by passing through a discrete memoryless channel. 

    Also related is the notion of local encodability, defined to be the maximum number of message symbols that influence any codeword symbol. Note that this is different from update efficiency, which is the number of codeword symbols that are influenced by any message symbol. Mazumdar and Pal~\cite{mazumdar2017semisupervised} observed the equivalence of locally encodable source coding with a problem of semisupervised clustering, and derived upper and lower bounds on the local encodability.
    Locality has been well studied in the context of channel coding---see, \emph{e.g.},~\cite{yekhanin2012locally,gopalan2012locality,tamo2014family,cadambe2015bounds,mazumdar2014update,tamo2016bounds} and the references therein. 
	
	 An outline of this paper is as follows. In Section~\ref{sec:problem_setup}, we describe the model. In Section~\ref{sec:main_results}, we present our results which are based on two schemes. The first achieves $(\ravg(1),\tavg(1))=(\Theta(1),\Theta(1))$ and the the second scheme achieves $(\rwc(1), \tavg(1))=(O(\log\log n), O(\log\log n)) $. The detailed description of these schemes as well as the proof of the main results appear in Sections~\ref{main} and~\ref{sec:o_loglogn_scheme}. In Section~\ref{conclusion}, we provide a few concluding remarks. 
	We end this section with notational conventions.

	\subsection*{Notation}\label{sec:notation}
	We use standard Bachmann-Landau notation for asymptotics. All logarithms are to the base $ 2 $. 	
	Curly braces denote sets, {\it{e.g.}}, $ \{a,b,c\}$, whereas parentheses are used to denote ordered lists, {\it{e.g.}}, $ (a,b,c) $.
	The set $ \{1,2,\ldots,n\} $ is denoted by $ [n] $, whereas for any positive integers $ i,m $, we define $ i:i+m$ to be $ \{i,i+1,\ldots,i+m  \} $.  The set of all finite-length binary sequences is denoted by $ \{0,1\}^* $.
	
	Random variables are denoted by uppercase letters, {\it{e.g.}}, $ X,Y $. Vectors of length $ n $ are indicated by a superscript $ n $, {\it{e.g.}}, $ x^n,y^n $. The $ i $th element of a vector $ x^n $ is $ x_i $. Uppercase letters with a superscript $ n $ indicate $ n $-length random vectors, {\it{e.g.}}, $ X^n,Y^n $. A substring of a vector $ x^n $ is represented as $ x_i^{i+m}\defeq (x_i,x_{i+1},\ldots,x_{i+m}) $.
	
	Let $ \cX $ be a finite set. For any $ x^n\in\cX^n $, let $ \hat{p}_{x^n} $ be the type/histogram  of $ x^n $, {\it{i.e.}},  $ \hat{p}_{x^n}(a)=\frac{\sum_{i=1}^n1_{\{x_i=a\}}}{n} $. We say that $ x^n $ is $ \varepsilon $-typical with respect to a distribution $ p_X $ if for all $ a\in\cX $, we have $ |\hat{p}_{x^n}(a)-p_X(a)|\leq \varepsilon p_X(a) $. Let $ \cT_{\varepsilon}^n $ denote the set of all $ n $-length sequences that are $ \varepsilon $-typical with respect to $ p_X $. 
	We impose an ordering (which may be arbitrary) on $ \cT_{\varepsilon}^n$.  If $ x^n\in \cT_\varepsilon^n $ is the $ i $th sequence in $ \cT_{\varepsilon}^n  $ according to the order, then we say that the index of $ x^n $ in $ \cT^{n}_\varepsilon $ (denoted by $ \mathtt{index}(x^n; \cT_\varepsilon^n) $) is $ i $.

	\section{Querying and updating compressed data}\label{sec:problem_setup}
	The source is specified by a distribution $p_X$ over a finite alphabet $ \cX $. Unless otherwise mentioned, a source sequence or a message refers to $n$ i.i.d. realizations $X^n$ of the source. 
	
	\begin{definition}[Compression scheme]
	A rate $R$ length $n$ compression scheme, denoted as $(n,R)$ compression scheme or $(n,R)$ fixed-length compression scheme, is a pair of maps $ (\enc,\dec) $ consisting of 
	\begin{itemize}
		\item An encoder $ \enc: \cX^n\to \{0,1\}^{nR} $, and
		\item A decoder $ \dec: \{0,1\}^{nR}\to \cX^n $.
	\end{itemize}
	 The probability of error is the probability of the event that codeword $\enc(X^n)$ is wrongly decoded, that is
	\[
	P_e \defeq \Pr_{X^n}[\dec(\enc(X^n))\neq X^n].
	\]
    \end{definition}
    

	
	\subsection{Queries}
	Given a compression scheme, a local decoder is an algorithm which takes $ (i,s)\in[n]^2 $ as input, adaptively queries (a small number of) bits of the compressed sequence $ C^{nR} $ and outputs $ X_{i}^{i+s-1} $.
	
	Given $ s\in [n] $ and codeword $ c^{nR} $ corresponding to source sequence $ x^n $, let $ d^{(s)}(i,x^n) $ denote the number of symbols of $ c^{nR} $ that need to be queried by the local decoder in order to decode $ x_i^{i+s-1} $ without error. The average local decodability of the code is defined as 
	\[
	\ravg(s)\defeq \max_{i\in [n-s+1]}\bE[d^{(s)}(i,X^n)],
	\]
	where the average is taken over $ X^n $ and possibly any randomness in the query algorithm. Hence, if say $\ravg(3)=20$ then the local decoder that can recover any length $3$ contiguous substring of the source by probing on average $20$ symbols from the codeword sequence. 
	
	The worst-case local decodability is defined as
	\[
	\rwc(s)\defeq \max_{i,x^n}d^{(s)}(i,x^n).
	\]

	\subsection{Updates}
	Given $ s\in [n] $, suppose a subsequence $x_i^{i+s-1}$ of the original sequence $ x^n $ is updated to $ \tilde{x}_i^{i+s-1} $ so that $x^n$ becomes $x^{i-1}\tilde{x}_i^{i+s-1}x_{i+s}^n$.
	A local updater is an algorithm which takes $ (i,\tilde{x}_i^{i+s-1}) $ as input, probes (a small number of) bits of the compressed sequence $ c^{nR} $, and modifies a small number of bits of $ c^{nR} $ such that the new codeword $ \tilde{c}^{nR} $ corresponds to the message $x^{i-1}\tilde{x}_i^{i+s-1}x_{i+s}^n$.  We assume here that the update algorithm probes and modifies  $ c^{nR} $ given $ (i,\tilde{x}_i^{i+s-1}) $ only, without prior knowledge of $ (x^n,c^{nR}) $. 
	
Accordingly,
	let $ \tread^{(s)}(i,x^n,\tilde{x}_i^{i+s-1}) $ and $\twrite^{(s)}(i,x^n,\tilde{x}_i^{i+s-1})$ denote the number of symbols of $ c^{nR} $ that need to be read and modified, respectively, and let $$ u_{\mathrm{tot}}^{(s)}(i,x^n,\tilde{x}_i^{i+s-1})\defeq \tread^{(s)}(i,x^n,\tilde{x}_i^{i+s-1})+\twrite^{(s)}(i,x^n,\tilde{x}_i^{i+s-1}).$$
	
	The average update efficiency of the code is defined as
	\[
	\tavg(s)\defeq \max_{i\in [n-s+1]}\bE\left[u_{\mathrm{tot}}^{(s)}(i,X^n,\tilde{X}_i^{i+s-1})\right]
	\]
	where the update $\tilde{X}_i^{i+s-1}$ is supposed to be independent of the original sequence $X^n$ but is drawn from the same i.i.d.$\sim p_X$ distribution. Hence, updates do not modify the distribution of the original message. The worst-case update efficiency is defined as
	\[
	\twc(s) \defeq \max_{i,x^n,\tilde{x}_{i}^{i+s-1}}u_{\mathrm{tot}}^{(s)}(i,x^n,\tilde{x}_i^{i+s-1}).
	\]
	
This paper is concerned about the design of $(n,H(p_X)+\varepsilon)$ compression schemes with vanishingly small probability of error that allows the recovery and update of short fragments (contiguous symbols) of the message efficiently. 	

	\section{Main results}\label{sec:main_results}
	A naive approach to achieve compression with locality is to partition the message symbols into nonoverlapping blocks of equal size $b$ and compress each block separately with a 
 $ (b,H(p_X)+\varepsilon) $ fixed-length compression scheme. The probability of error for each block can be made to go to zero as $ 2^{-\Theta(b)} $ (see, {\it{e.g.}},~\cite{cover2012elements}). From the union bound, the overall probability of error is at most $ (n/b)2^{-\Theta(b)} $. Hence, as long as $ b=\Omega(\log n) $ we have $ P_e=o(1) $.  
Since the blocks are encoded and decoded independently, $$ \rwc(1)=\twc(1)=O(b)= O(\log n) $$
where the constant in the order term does not depend on $\varepsilon$.
 The overall computational complexity is at most $(n/b)2^{\Theta(b)}$, which is polynomial in $n$.
Noticing that every subsequence of length $s> 1$ is contained in at most $\lceil s/b \rceil+1$ blocks, we have:\footnote{In case $b$ does not divide $n$, we can compress the last block of size $b+n-\lfloor n/b\rfloor b$ separately using a ($b+n-\lfloor n/b\rfloor b, H(p_X)+\varepsilon$)-fixed length compression scheme. The local decodability and update efficiency would increase by a factor of less than $2$, and therefore remain $O(\log n)$. A similar argument can be made for all the multilevel schemes in the rest of this paper and overall will only introduce an additional constant multiplicative factor. For ease of exposition, we will conveniently assume in all our proofs that the size of each block divides $n$.}
	\begin{theorem}[Fixed-length neighborhood and compression]\label{naive}
		For every $ \varepsilon>0 $, the naive scheme achieves a rate-locality triple of
		\[
		(R,\rwc(1),\twc(1)) = \left( H(p_X)+\varepsilon, O(\log n),O(\log n) \right) .
		\]
		\label{lemma:Ologn}
		Moreover,
		\[
		\rwc(s) = \begin{cases}
		\Theta(\log n), & \text{if }s\leq b\\
		\Theta(s), &\text{if } s>b
		\end{cases}
		\]
		\[
		\twc(s) = \begin{cases}
		\Theta(\log n), & \text{if }s\leq b\\
		\Theta(s), &\text{if } s>b
		\end{cases}
		\]
		where all the order terms are independent of $\varepsilon$. The overall computational complexity required for compression/decompression is polynomial in $n$.
	\end{theorem}
	It is easy to see that the above analysis is essentially tight as the naive scheme achieves vanishingly small error probabilities for overall compression and decompression only if $ b=\Omega(\log n) $.
	
	In the naive scheme, the recovery or update of a particular symbol $X_i$ involves an $O(\log n)$-size neighborhood of that symbol which is compressed by means of a fixed-length compression scheme. To improve upon the $O(\log n)$ locality, we consider two other schemes. In the first, neighborhoods are of variable lengths and are compressed using a fixed length block code. 
	The length of the neighborhood of a particular symbol $X_i$ is defined as the length of the smallest typical set that contains $X_i$. To find this smallest neighborhood, the algorithm proceeds iteratively by considering larger and larger neighborhoods of $X_i$ until it finds a neighborhood that is typical. Local decoding and local recovery of $X_i$ are performed by decompressing and recompressing this neighborhood. 
This scheme is formally described in Section~\ref{main} where we prove the following result:
	\begin{theorem}[Variable-length neighborhood and fixed length compression]\label{varlen}
		Fix $ \varepsilon>0 $.
		There exists a scheme which universally over i.i.d. sources with common known finite alphabet achieves rate
		$R = H(p_X)+\varepsilon,$
		and  probability of error 
		$
		\Pr[\dec(\enc(X^n))\neq X^n] = 2^{-2^{\Omega(\sqrt{\log n})}}.
		$
		The average local decodability and update efficiency is 
		\[
		\ravg(s)\leq\begin{cases}
		\alpha_1 \frac{1}{\varepsilon^2}\log \frac{1}{\varepsilon} &\text{if } s \leq \alpha_1''\left(\frac{1}{\varepsilon^2}\log \frac{1}{\varepsilon}\right)\\
		\alpha_1's &\text{if } s >\alpha_1''\left(\frac{1}{\varepsilon^2}\log \frac{1}{\varepsilon}\right)
		\end{cases},
		\]
		\[
		\tavg(s)\leq\begin{cases}
		\alpha_2 \frac{1}{\varepsilon^2}\log \frac{1}{\varepsilon} &\text{if } s \leq \alpha_2''\left(\frac{1}{\varepsilon^2}\log \frac{1}{\varepsilon}\right)\\
		\alpha_2's &\text{if } s > \alpha_2''\left(\frac{1}{\varepsilon^2}\log \frac{1}{\varepsilon}\right)
		\end{cases},
		\]
		where the constants $ \alpha_i,\alpha_i',\alpha_i'' $, $i=1,2$, are  independent of $ n,\varepsilon $ but dependent on $ p_X $.
		Moreover, the overall computational complexity of encoding and decoding $ X^n $ is $ O(n\log n) $. For $1\leq s\leq n$, the expected computational complexity for local decoding or updating a fragment of size $s$ is $\Theta(s)$, where the proportionality constant depends only on $\varepsilon$ and $p_X$.\footnote{In comparison, the naive scheme requires computational complexity $\Omega(\log n)$ to locally decode or update even a single symbol.}
		\label{theorem:main_iid}
	\end{theorem}
	
	Mazumdar {\it{et al.}}~\cite{mazumdar2015local} proved that $\ravg^*(1)=\Omega(\log (1/\varepsilon))$ for non-dyadic sources.\footnote{Recall that $ \ravg^*(1) $ denotes the minimum average local decodability that can be achieved by any compression scheme having rate $ R \leq H(p_X)+\varepsilon $.} Hence, from Theorem~\ref{theorem:main_iid} we get:
	\begin{corollary}\label{coro}
There exists a universal constant $\alpha_{f}>0$ such that for all non-dyadic sources, the scheme of Theorem~\ref{varlen} achieves $\ravg(s)<s\ravg^*(1)$ whenever $s\geq \alpha_f/\varepsilon^2$.
	\end{corollary}
Given Theorem~\ref{naive}, the interesting regime of Corollary~\ref{coro} is when attempting to locally decode a substring of size $s$ that satisfies $$\Omega(1/\varepsilon^2)\leq s\leq o(\log n). $$

Theorem~\ref{varlen} involves average local decoding and average local update. A natural question is whether we can achieve the same performance but under worst-case locality, {\it{i.e.}}, can we achieve for any $1\leq s\leq n$ $$(\rwc(s),\twc(s))=(O(s),O(s))?$$
While this question remains open we show that it is possible to achieve $(\rwc(s),\tavg(s))=(O(s)),O(s))$ whenever $s=\Omega(\log \log (n))$.
This result is obtained by means of a second scheme where neighborhoods are of fixed length, as in the naive scheme, but compressed with a variable length code. 
 Using such as a code raises the problem of efficiently encoding the start and end locations of each subcodeword. Indeed, were we to store an index of the locations of each subcodeword, and since there are $ n/b $ subcodewords, the index would take approximately $ (n\log n)/b $ additional bits of space. Hence, only to ensure that the rate remains bounded would require $ b=\Omega(\log n) $, which would further imply that $ \rwc(1) $ and $ \twc(1) $ are still $ O(\log n) $. It turns out that the location of individual subcodewords can be done much more efficiently by means of a particular data structure for subcodeword location as we show in Section~\ref{sec:o_loglogn_scheme}:
	\begin{theorem}[Fixed-length neighborhood and variable-length compression]\label{fixlen}
		Fix $ \varepsilon>0 $. There exists a scheme which univerally over i.i.d.\ sources with common known finite alphabet achieves a rate-locality triple of
		\[
		(R,\rwc(1),\tavg(1)) = (H(p_X)+\varepsilon,O(\log\log n),O(\log\log n))
		\] 
		where order terms are independent of $\varepsilon$.
		
		Moreover, for any $ s > 1 $, 
		\[
		\rwc(s) \leq  \begin{cases}
		2\rwc(1) &\text{if }s\leq b_1\\
		s(H(p_X)+\varepsilon)+2\rwc(1) &{otherwise},
		\end{cases}
		\]
		and
		\[
		\tavg(s) \leq  \begin{cases}
		2\tavg(1) &\text{if }s\leq b_1\\
		2s(H(p_X)+\varepsilon)+2\tavg(1) &\text{otherwise},
		\end{cases}
		\]
		where $ b_1 =O(\log\log n).$
		The overall computational complexity of encoding and decoding is polynomial in $n$.
		\label{thm:Ologlogn}
	\end{theorem}

Analogously to the derivation of Corollary~\ref{coro} we get:
	\begin{corollary}\label{coro2}
For non-dyadic sources, there exists a constant $\alpha_v>0$ such that the scheme of Theorem~\ref{fixlen} achieves  $\rwc(s)< s\rwc^*(1) $ whenever $s\geq \alpha_v(\log\log n)$.
	\end{corollary}

All our results easily extend to variable-length codes with zero error---See Appendix~\ref{sec:conversion_fixlength_varlength}.

	\subsection*{Discussion}
	
	Mazumdar {\it{et al.}}~\cite{mazumdar2015local} gave a compression scheme that achieves $ R=H(p_X)+\varepsilon $ and $ \rwc(1)=\Theta(\frac{1}{\varepsilon}\log \frac{1}{\varepsilon}) $. 
	The probability of error decays as $ 2^{-\Theta(n)} $. This suggests that we can achieve $ \twc(1)=O(\log n) $ using the following scheme. Split the message into blocks of $ O(\log n) $ symbols each, and use the scheme of Mazumdar {\it{et al.\ }} in each block. We can choose the size of each block so that the overall probability of error decays polynomially in $ n $. Since each block of size $ O(\log n) $ is processed independently of the others, the overall computational complexity (which may be exponential in the size of each block) is only polynomial in $ n $. This gives us the following result:
	\begin{lemma}[Corollary to~\cite{mazumdar2015local}]
		For every $ \varepsilon>0 $, a rate-locality triple of 
		$$ (R,\rwc(1),\twc(1))= \left(H(p_X)+\varepsilon, \Theta\left(\frac{1}{\varepsilon}\log \frac{1}{\varepsilon}\right), O(\log n) \right) $$
		is achievable with $ \mathrm{poly}(n) $ overall  encoding and decoding complexity.
		\label{lemma:mazumdar_result} 
	\end{lemma} 
	Although the above scheme has $ \mathrm{poly}(n) $ computational complexity, this could potentially be a high-degree polynomial.
	Moreover, we do not know if the above scheme can achieve $ \rwc(s)< s\rwc(1) $ for $ 1<s = o(\log n) $.
	
	Montanari and Mossel~\cite{montanari2008smooth} gave a compressor that achieves update efficiency $ \twc(1) =\Theta(1)$. The construction is based on syndrome decoding using low-density parity-check codes. Arguing as above we deduce the following lemma:
	\begin{lemma}[Corollary to~\cite{montanari2008smooth}]
		For every $ \varepsilon>0 $, a rate-locality triple of 
		$$ (R,\rwc(1),\twc(1))= \left(H(p_X)+\varepsilon,  O(\log n), \Theta(1) \right) $$
		is achievable with $ \mathrm{poly}(n) $ overall  encoding and decoding complexity.
		\label{lemma:montanari_result} 
	\end{lemma}
	The local decodability of the compressor in~\cite{montanari2008smooth} cannot be improved as it uses a linear encoder for the compression of each block, and Makhdoumi {\it{et al.}}~\cite{makhdoumi_onlocallydecsource} showed that for such a compression scheme local decodability ($\rwc(1)$) necessarily scales logarithmically with block size, hence in our case $ \rwc(1)=\Omega(\log n) $. Hence, linearity in the encoding impacts local decodability.
 Interestingly, they also noted that if we impose the decoder to be linear then it is impossible to achieve nontrivial rates of compression irrespective of $\rwc(1)$. 

	\section{Proof of Theorem~\ref{theorem:main_iid}}\label{main}
	We now present our compression scheme which achieves constant $ (\ravg(1),\tavg(1)) $. We assume first that the source distribution $p_X$ is known, as it is conceptually simpler. The universal scenario is handled separately in Section~\ref{sec:scheme_lz}.
	
    \begin{figure}
        \centering
        \includegraphics[width=11cm]{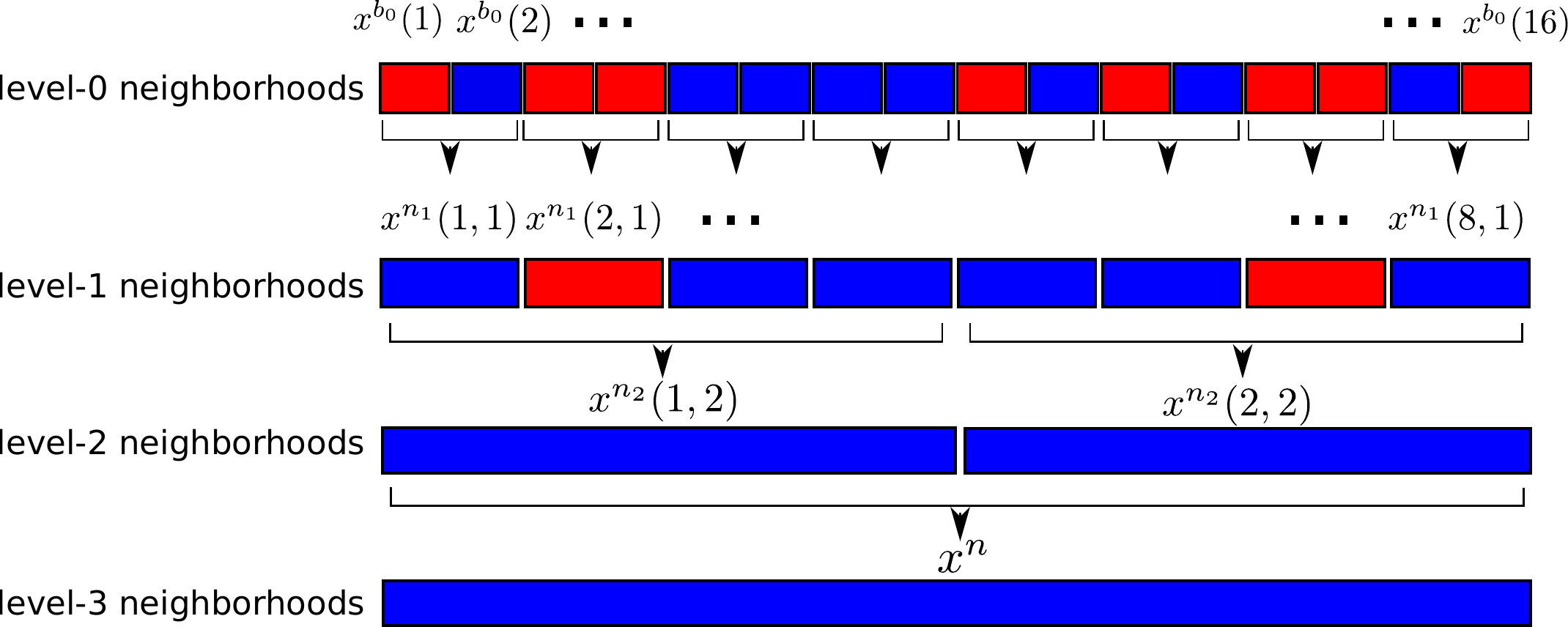}
        \caption{Intution for the multilevel compression scheme in Section~\ref{sec:scheme_constant}. We group together symbols to form larger neighborhoods. If we have an efficient means to compress these neighborhoods, then we can locally decode a block by decompressing the smallest typical neighborhood of that block. Blocks colored blue are typical, while the red blocks are atypical.}
        \label{fig:multilevel_compression_intuition}
    \end{figure}
    Before giving a formal description of our scheme, let us give some intuition.
    \subsection{Intuition}
     The main idea is to analyze the message sequence at multiple levels: At the coarsest level, we view the message as a single block of size $n$. At the finest level, we view it as a concatenation of blocks of size $b_0=\Theta(1)$. As depicted in Figure~\ref{fig:multilevel_compression_intuition}, we can refine this by saying that at level $\ell$, the message is viewed as a concatenation of $n_\ell$-sized neighborhoods, where $b_0<n_1<n_2<\ldots<n$. If $b_0=\Theta(1)$, then a positive fraction of the level-$0$ neighborhoods are atypical with high probability, while neighborhoods at higher levels are more likely to be typical. 
    
    Corresponding to each $b_0$-sized block, we identify the smallest typical neighborhood containing the block. In the example of Figure~\ref{fig:multilevel_compression_intuition}, the smallest typical neighborhood of $x^{b_0}(1)$ is $x^{n_1}(1,1)$ at level $1$, while that of $x^{b_0}(2)$ is $x^{b_0}(2)$ itself. 
    The main idea in our scheme is to efficiently encode typical neighborhoods at each level, and local decoding/update of a symbol is performed by decompressing/recompressing only the smallest typical neighborhood containing it.

    Our actual scheme is more nuanced. We will view the message sequence at different levels, but use a different definition of typicality at  each level. Compression of the neighborhoods is performed in an iterative fashion, starting from level $0$, and then moving to higher levels. At level $\ell,$ we only compress the residual information of each neighborhood, \emph{i.e.,} that which is not recoverable from the first $\ell-1$ levels. 
    
    Local decoding of a symbol is performed by successively answering the question ``Is the level-$\ell$ neighborhood typical?'' for $\ell=0,1,\ldots$, till we get a positive answer. The desired symbol can be recovered from the typical neighborhood.

    We proceed with the formal description of our scheme.
    
	\subsection{Compression scheme}\label{sec:scheme_constant}
	Fix $ \varepsilon_0>0 $. Let $ b_0 \defeq n_0\defeq \Theta(\frac{1}{\varepsilon_0^2}\log\frac{1}{\varepsilon_0}) $ where the implied constant is chosen so that $$ \Pr[X^{b_0}\notin \cT^{b_0}_{\varepsilon_0}]\leq \varepsilon_0^4, $$ 
	and let $$k_0\defeq \lceil (H(p_X)+\varepsilon_0)b_0\rceil.$$
	
	For $ \ell\geq 1 $,  let $$ \varepsilon_\ell =\varepsilon_{\ell-1}/2,$$
	$$b_\ell=4b_{\ell-1},$$
	$$n_\ell= b_\ell n_{\ell-1}, $$  and let $ \ell_{\max} $ be the largest $\ell$ such that 
	$
	n_\ell \leq n.
	$
	
	Notice that $ \ell_{\max}=\Theta(\sqrt{\log n}) $.  
	
	The overall encoding/decoding involves a multilevel procedure over $ \ell_{\max} $ levels. At each level, we generate a part of the codeword and modify the input string in an entropy decreasing manner until the string becomes a constant.
	The scheme uses a special marker symbol, referred to as 
	$ \diamond $, that is not in $\cX $. 
	This symbol will be used to denote that we have been able to successfully compress a part of the message at an earlier stage.
	\begin{definition}[$ \diamond $ blocks and non-$\diamond $ blocks]
		A vector $ v^m $  is said to be a $ \diamond $-block if $ v_i=\diamond  $ for all $ i $. It is called a non-$ \diamond $ block if there exists an $ i $ such that $ v_i\neq \diamond $.
	\end{definition}

	\begin{figure}
	    \centering
	    \includegraphics[width=4cm]{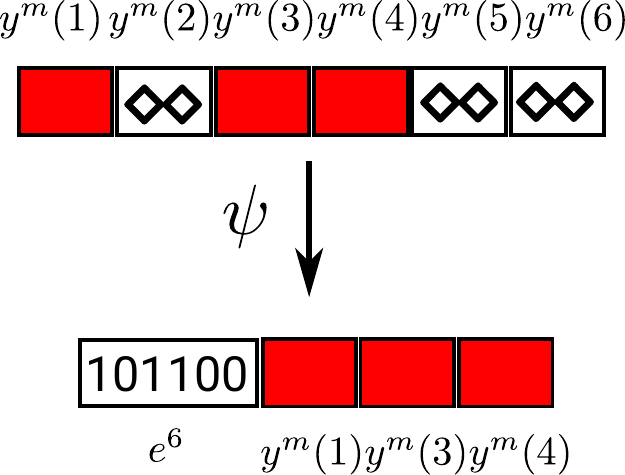}
	    \caption{Illustrating the compression scheme for levels $\ell\geq 1$ as described in Definition~\ref{remark:coding_at_level_i}. In this example, we have used $b=6$ and $\varepsilon=1/2$.}
	    \label{fig:higherlevel_compression}
	\end{figure}
	
	\begin{figure}
		\begin{center}
			\includegraphics[width=13cm]{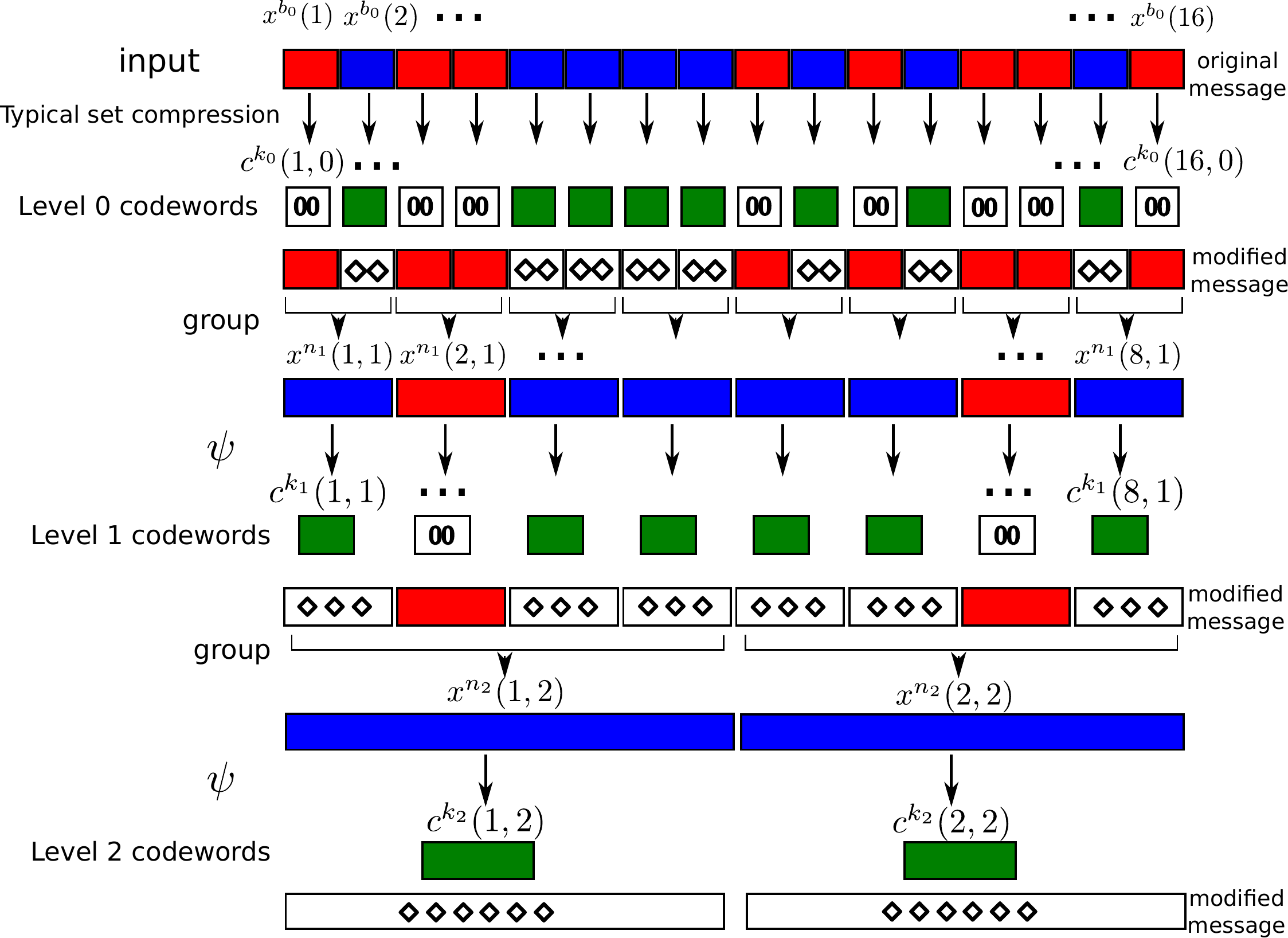}
			\caption{Illustrating the multilevel compression scheme. Red and blue blocks denote atypical and typical blocks respectively, while green blocks denote nonzero codewords. For ease of illustration, we have used $b_{\ell}=2b_{\ell-1}$.}
			\label{fig:multilevel_compression}
		\end{center}
	\end{figure}

	\subsubsection{{{Level $\ell=0$}}}
	partition $x^n$ into $ n/b_0 $ blocks of length $ b_0=n_0 $ each. Let $ x^{n_0}(j)\defeq x_{(j-1)n_0+1}^{jn_0} $ denote the $ j $th block of the message symbols. Blocks at level $\ell=0$ are processed independently of each other. For each $ x^{n_0}(j) $, we generate a codeword block $ c^{k_0}(j,0) $ and possibly modify $ x^{n_0}(j) $:
	\begin{itemize}
		\item If $ x^{n_0}(j) $ is typical, then $ c^{k_0}(j,0) $ is assigned the index of $ x^{n_0}(j) $ in $\cT_{\varepsilon_0}^{n_0}$, else  $c^{k_0}(j,0) = 0^{k_0} $. 
		\item If $ x^{n_0}(j) $ is typical, then  $x^{n_0}(j,0)$ is modified to a diamond block $\diamond^{n_0} $ and if  $ x^{n_0}(j) $ is not typical then $x^{n_0}(j,0)$ is kept unchanged. The message sequence after possible modifications of each block $x^{n_0}(j,\ell=0)$, $j=1,2,\ldots$ is denoted by $x^n(\ell=0)$.
		
	\end{itemize} 
	
	For compression at higher levels, we make use of the following code
	\begin{definition}[Code for levels $\ell
	\geq 1$]
		Fix any positive integers $ b,m $.
		Let $ \cX $ be a finite alphabet, and $ \diamond  $ be a symbol such that $ \diamond \notin \cX $. Let $\cS\subset (\cX\cup\{\diamond \})^{mb}$ be the set of all sequences of the form $ y^{mb} = (y^m(1),y^m(2),\ldots,y^m(b)) $ such that $ y^m(j)\in  (\cX\cup\{\diamond \})^{m}$  and  at least $ (1-\varepsilon) b $ fraction of the $ y^m(j) $'s are $ \diamond  $ blocks. 
		
		For any sequence $ y^{mb}\in\cS $, let $ j_1,j_2,\ldots,j_{k} $ denote the locations of the non-$\diamond $ blocks. Let $e^b =  \phi(y^{mb};b,m) $ be the $ b $-length indicator vector for the non-$\diamond $ blocks, {\it{i.e.}}, the $ j $th element of $ \phi(y^{mb};b,m) $ is $ 1 $ iff $ y^{mb}(j) $ is a non-$\diamond $ block.
		Let $$ \psi(y^{mb};b,m,\varepsilon) \defeq (e^b, y^{mb}(j_1),,\ldots,y^{mb}(j_k),\diamond ^{m(\varepsilon b-k)} ).  $$
		In other words, $\psi$ consists of a header $e^b$ to locate the non-$\diamond$ blocks, followed by a concatenation of all the non-$\diamond$ blocks.
		The binary representation of $ \psi $ requires $ b+\varepsilon mb\log(|\cX|+1) $ bits. The mapping $ \psi $ is one-to-one on $ \cS $. Both $ \psi $ and $ \psi^{-1}  $ (for any element in the range of $ \psi $) can be computed using $ \Theta(mb) $ operations. An example is illustrated in Figure~\ref{fig:higherlevel_compression}.

		\label{remark:coding_at_level_i}
	\end{definition}

	\subsubsection{{{Levels $\ell\geq 1$}}}
	having generated codewords up to level $\ell-1$ and having modified the message if necessary, we form groups of $ b_\ell $ consecutive blocks from $x^{n}(\ell-1)$ to obtain blocks of size $n_\ell=b_\ell n_{\ell-1}$. The $jth$ block at level $\ell$, denoted  $x^{n_\ell}(j,\ell)$, is therefore
	$$(x^{n_{\ell-1}}((j-1)b_\ell+1,\ell-1),\ldots,x^{n_{\ell-1}}(jb_\ell+1,\ell-1)).$$
	Similarly to level $\ell=0$, for each of these blocks of size $n_\ell$, we generate a codeword and modify it if necessary:  
	\begin{itemize}
		\item If $x^{n_\ell}(j,\ell)$ is ``typical,'' {\emph{i.e.}}, has at least  $ (1-\varepsilon_\ell)b_\ell $ $ \diamond $-blocks (of size $ n_{\ell-1} $), then we set the subcodeword $ c^{k_\ell}(j,\ell) $ of length $k_\ell= b_\ell+\varepsilon_\ell n_\ell\log(|\cX|+1) $ using the scheme described in Definition~\ref{remark:coding_at_level_i}.\footnote{One could use a more sophisticated scheme to get better performance. However, we can get order-optimal $ (\ravg,\tavg) $ even with this very simple scheme.} If this block is ``atypical,'' {\it{i.e.}}, has fewer than $(1-\varepsilon_\ell)b_\ell $ many $\diamond$ blocks, then $ c^{k_\ell}(j,\ell)=0^{k_\ell} $.
		\item If $x^{n_\ell}(j,\ell)$ has at most $ \varepsilon_\ell b_\ell $ many non-$ \diamond $-blocks, then we modify $x^{n_\ell}(j,\ell)$ to a diamond block $\diamond^{n_\ell} $. Otherwise, the group is left untouched.
	\end{itemize}
	Hence, at each level the input sequence gets updated with more and more $ \diamond $'s as larger and larger subsequences become typical. As we show in Section~\ref{analysis}, the entropy of the message keeps decreasing till it becomes zero, once it becomes the all-$ \diamond $ sequence. 
	Finally, the stored codeword is the concatenation of codewords of all levels:
	\[
	c^{nR} = (c^{k_0}(1:n/n_0,0),\ldots,c^{k_{\ell_{\max}}}(1:n/n_{\ell_{\max}},\ell_{\max})).
	\]


	
	\begin{example}[Figure~\ref{fig:multilevel_compression}]
	 An example of the encoding process is illustrated in Figure~\ref{fig:multilevel_compression} where
	 the blue blocks refer to typical blocks whereas the red blocks refer to atypical blocks. 
	 
	 At level $0$, the subcodewords $c^{k_0}(i,0)$ are obtained using typical set compression. The subcodeword $c^{k_0}(i,0)$ is zero if the block is atypical, and nonzero (depicted in green in the figure) if it is typical. We then modify the message, replacing each typical level-$0$ block with $\diamond^{b_0}$. 
	 
	 For ease of illustration, we select $b_1=2$ and $\varepsilon_1=1/2$. Hence the blocks are grouped in pairs to obtain $x^{n_1}(i,1)$, $1\leq i\leq 8$. A block $x^{n_1}(i,1)$ is typical if it contains at most one non-$\diamond$ block of length $n_0$. Therefore, only $x^{n_1}(2,1)$ and $x^{n_1}(7,1)$ are atypical. These blocks are compressed to get the level-$1$ codewords $c^{k_1}(i,1)$ for $1\leq i\leq 8$. As earlier, typical blocks are encoded to nonzero codewords, while atypical blocks are compressed to the zero codeword. Post compression, we again modify the message by replacing typical blocks with $\diamond^{n_1}$.
	 
	 The encoding process proceeds in an identical fashion for level $2$, where we have selected $b_2=4$ and $\varepsilon_2=1/4$.
	 \end{example}
	
	
	\subsection{Local decoding}
	Suppose that we are interested in recovering the $ m $th message symbol $ x_m $, where $m\in (j-1)n_0:jn_0 $. 
	\begin{itemize}
		\item We probe $ c^{k_0}(j,0) $. If the block $ x^{n_0}(j) $ is typical, then we can directly recover $x^{n_0}(j)$ from $ c^{k_0}(j,0) $.
		\item If $ x^{n_0}(j) $ is not typical, we probe higher levels successively till we reach the smallest level $ \ell $ for which the block that includes  $x^{n_0}(j)$, which we denote as  $ x^{n_\ell}(q_{\ell}(j),\ell) $ is a diamond $\diamond^{n_\ell}$- block. This can be determined by reading the first $ b_i $ bits of $ c^{k_i}(q_{i}(j),i) $, $i=1,2,\ldots, \ell$ since this corresponds to the indicator vector of the non-$ \diamond $ blocks at each level $i\leq \ell $. If we can recover $ x^{b_0}(j) $ by probing up to the first $ \ell $ levels, then we  say that the $ j $th block is encoded at the $ \ell $th level.
		\item Using this approach, we automatically recover the entire block $ x^{b_0}(j) $---not only an individual message symbol. If we want to recover multiple message blocks, we repeatedly employ the same algorithm on each block.\footnote{We can actually do much better than naively repeating the algorithm for multiple blocks. However, for ease of exposition and proofs, we use the naive algorithm.} 
	\end{itemize}
	
		We revisit our earlier example to illustrate the local decoder.
\begin{example}[Figure~\ref{fig:multilevel_compression}]
	Suppose that we are interested in recovering $x^{b_0}(2)$. The local decoder first probes $c^{k_0}(2,0)$. Since this is a nonzero codeword, $x^{b_0}(2)$ can be obtained by decompressing $c^{k_0}(2,0)$. In this process, the local decoder probes $k_0$ bits.
	
	Suppose that we are instead interested in recovering $x^{b_0}(3)$. On probing $c^{k_0}(3,0)$, the local decoder obtains a zero codeword. Next, it probes $c^{k_1}(2,1)$. This is also zero. Finally, the local decoder probes $c^{k_2}(1,2)$ which is nonzero, and $x^{b_0}(3)$ can be obtained by decompressing this codeword. In this case, the local decoder probes $k_0+k_1+k_2$ bits.
	\end{example}
	\subsection{Local updating}
	The local updating rule is a little more involved. Assume that the $ j $th block $ x^{n_0}(j) $ is to be updated with $ \tilde{x}^{n_0}(j) $.
	\begin{itemize}
		\item If both $ x^{n_0}(j) $ and $ \tilde{x}^{n_0}(j) $ are typical, only $ c^{n_0}(j,0) $ needs to be updated. Whether $ x^{n_0}(j) $ is typical or not can be determined by reading  $ c^{k_0}(j,0) $. 
		\item If both $ x^{n_0}(j) $ and $ \tilde{x}^{n_0}(j) $ are atypical, then we probe higher levels till we reach the level $ \ell $ where $ x^{n_0}(j) $ is encoded, and update $ c^{k_0}(q_{\ell}(j),
		\ell) $.
		\item If $ x^{b_0}(j) $ is typical and $ \tilde{x}^{b_0}(j) $ is atypical, then we need to update $ c^{b_0}(j,0) $ and the blocks at higher levels. Due to the atypicality, the number of non-$ \diamond $ blocks for level $ 1 $ increases by $ 1 $, and hence $ c^{\ell_1}(q_1(j),1) $ must be updated. If the number of non-$ \diamond $ blocks now exceeds $ \varepsilon_1b_1 $, then we would also need to update the codeword at level $ 2 $, and so forth.
		\item If $ x^{b_0}(j) $ is atypical and $ \tilde{x}^{b_0}(j) $ is typical, then the number of non-$ \diamond $ blocks at each level might decrease by $ 1 $ (or $ 0 $). If $ x^{b_0}(j) $ were encoded at level $ i $, then we might need to update the codeword blocks up to level $ i $.  
	\end{itemize}
	Let us illustrate the local updater in the context of our earlier example.
	\begin{example}[Figure~\ref{fig:multilevel_compression}]

	Suppose that we want to replace $x^{b_0}(5)$ with $\widetilde{x}^{b_0}(5)$. The local updater first probes $c^{k_0}(5,0)$ to conclude that $x^{b_0}(5)$ is encoded at level $0$. 
	
	If $\widetilde{x}^{b_0}(5)$ is also typical, then only $c^{k_0}(5,0)$ needs to be updated, and the rest of the codeword remains untouched. The updater probes $k_0$ bits and modifies $k_0$ bits. 
	
	In case $\widetilde{x}^{b_0}(5)$ is atypical, then the local updater first sets $c^{k_0}(5,0)$ to $0^{k_0}$. It then probes $c^{k_1}(3,1)$ and decompresses this to recover $x^{n_1}(3,1)$. This block is updated with $\widetilde{x}^{b_0}(5)$, and the new level $1$ block $\widetilde{x}^{n_1}(3,1)$ is typical. Therefore, $c^{k_1}(3,1)$ is updated with the codeword corresponding to  $\widetilde{x}^{n_1}(3,1)$, and the update process is terminated. In this scenario, the updater probes $k_0+k_1$ bits and modifies $k_0+k_1$ bits.
	\end{example}
	
	\subsection{Connections with P\u{a}tra\c{s}cu's compressed data structure~\cite{patrascu2008succincter}}\label{sec:patrascu-connections}
    
    In~\cite{patrascu2008succincter}, P\u{a}tra\c{s}cu gave an entropy-achieving compression scheme that achieves constant-time local decoding in the word-RAM model. The compressor has a multilevel structure whose concept inspired our work. 
    
    The basic idea in~\cite{patrascu2008succincter} is the following. At level $0$, split the message into blocks of $b_0$ symbols each, compress each block using an entropy-achieving variable-length compression scheme, and store a fixed number of the compressed bits of each block. The remainder is called the ``spill,'' and is encoded in higher levels. At level $i\geq 1$, the spills from each block of level $i-1$ are grouped together to form larger blocks, and compressed in a fashion similar to level $0$.  Reconstruction of any block necessarily requires both the codeword at level-$0$ and the spill. As a result, the local decoder of~\cite{patrascu2008succincter} must always probe subcodewords of all levels, and the number of bitprobes required to recover even one symbol is $\Omega(\log n)$.
    
    In our scheme on the other hand encoding is such that the number of levels that the local decoder needs to probe to retrieve one block depends on the realization of the source message. 
    In particular, the local decoder need not always probe all levels---and indeed, probes a small number of levels. 
    
    Hence, in P\u{a}tra\c{s}cu's scheme the information about a particular block is spread across multiple levels whereas in our scheme this information is stored at a particular level that depends on the realization of the message.

	In the next section we establish Theorem~\ref{theorem:main_iid} assuming the underlying source $p_X$ is known. Universality is handled separately in Section~\ref{sec:scheme_lz}. 
	
\subsection{Bounds on $ \ravg(1) $ and $ \tavg(1) $}\label{analysis}
	
	We now derive bounds on the average local decodability and update efficiency. In the following, we will make use of some preliminary results that are derived in Appendix~\ref{sec:preliminary_lemmas}.
	\begin{lemma}
		If $\varepsilon_0<1/2$, then
		\[
		\ravg(1) \leq 2b_0.
		\]
		\label{lemma:localdecoder_expectedbits}
	\end{lemma}
	\begin{proof}
		We can assume without loss of generality that we want to recover $ X_1 $.
		
		If $ X_1 $ is encoded at level  $ i $, then the local decoder probes $ \sum_{i_1=0}^i k_{i_1} $ bits.
		Therefore,
		\begin{align}
		\bE[\ravg^{(1)}(X^n,b_0)] &\leq b_0+\sum_{i=1}^{\ell_{\max}}\left(\Pr[ X_1\text{ is encoded at level }i]\sum_{i_1=0}^{i}k_{i_1}\right)&\notag\\
		&\leq b_0+\sum_{i=1}^{\ell_{\max}}\left(\Pr[ X_1\text{ is encoded at level }i]\sum_{i_1=0}^{i}n_{i_1}\right)&\notag\\
		&\leq b_0+\sum_{i=1}^{\ell_{\max}}\Big((i+1)n_{i} \Pr[ X_1\text{ is encoded at level }i]\Big). &\label{eq:localdecoder_expectedbits_proof_1}
		\end{align}
		Let $ \delta^{(1)}_{i\to i+1} $ denote the conditional probability that $ x^{n_i}(1,i) $ is not the all-$\diamond$  block given that $ x^{n_{i-1}}(1,i-1) $  is not a $ \diamond  $-block.
		Then,
		\begin{align*}
		\Pr[X_1\text{ is encoded at level }i] &\leq \delta_{0\to 1}\prod_{i_1=1}^{i-1}\delta_{i_1\to i_1+1}^{(1)}
		\end{align*}
		 From Lemma~\ref{lemma:prob_leftover_condn}, specifically (\ref{eq:prob_leftover_condn_1}), we know that $\delta_{i\to i+1}^{(1)} \leq \varepsilon_i^{\beta 2^{i-1}}$ for $ i\geq 1 $. The quantity $ \beta $ is defined in \eqref{eq:defn_beta}.
		Therefore, 
		\begin{align}
		\Pr[X_1 \text{ is encoded at level }i] &\leq \varepsilon_i^{\beta 2^{i-1}}. &\label{eq:localdecoder_expectedbits_proof_2}
		\end{align}
		Since $ n_{i_1} = b_0^{i_1+1}2^{i_1(i_1+1)} $, we have 
		 \[
		 (i+1)n_i  \leq (i+1) b_0^{i+1}2^{i(i+1)}.
		 \]
		 Using this and \eqref{eq:localdecoder_expectedbits_proof_2} in~\eqref{eq:localdecoder_expectedbits_proof_1}, we have
		 \begin{align}
		 \bE[\ravg^{(1)}(X^n,b_0)] &\leq b_0+ \sum_{i=1}^{\ell_{\max}} (i+1) b_0^{i+1}2^{i(i+1)} \varepsilon_i^{\beta 2^{i-1}}. 
		 \end{align}
		It is easy to show that $ (i+1)b_0^{i+1}2^{i(i+1)} \varepsilon_i^{\beta 2^{i-1}} \leq \varepsilon_0^{i}$ for all $ i\geq 1 $ (see Lemma~\ref{lemma:locdec_intermediate_bound} for a proof). Therefore,
		\[
		\ravg(1) = \bE[\ravg^{(1)}(X^n,b_0)] \leq b_0+\varepsilon_0b_0\sum_{i=1}^{d_{\max}}\varepsilon_0^{i} <2b_0
		\]
		if $\varepsilon_0<1/2$.
		This completes the proof.
	\end{proof}
	
	\begin{lemma}
		If $\varepsilon_0<1/2$, then
		\[
		\tavg(1) \leq 8b_0.
		\]
		\label{lemma:updating_expectedbits}
	\end{lemma}
	\begin{proof}
		The calculations are identical to those in Lemma~\ref{lemma:localdecoder_expectedbits}, so we will only highlight the main differences. Again, we can assume that the first symbol needs to be updated.
		
		Suppose $ U^{b_0}(1) $ is the new realization of the message block that needs to be updated.
		Let $ i_{\mathrm{old}} $ denote the level at which $ X^{b_0}(1) $ is encoded in the codeword for $ X^n $, and let $ i_{\mathrm{new}} $ be the level at which $ U^{b_0}(1) $ is encoded in the codeword for $U^{b_0}(1),X^{(b_0)}(2),\ldots,X^{b_0}(n/b_0)$. The number of bits that need to be read is upper bounded by
		$$ u_{rd} \leq \max\{(i_{\mathrm{old}}+1)n_{i_{\mathrm{old}}},(i_{\mathrm{new}}+1)n_{i_{\mathrm{new}}}\}\leq (i_{\mathrm{old}}+1)n_{i_{\mathrm{old}}}+(i_{\mathrm{new}}+1)n_{i_{\mathrm{new}}}.$$
		Likewise, the number of bits that need to be written is  $$ u_{wr} \leq  \max\{(i_{\mathrm{old}}+1)n_{i_{\mathrm{old}}},(i_{\mathrm{new}}+1)n_{i_{\mathrm{new}}}\}\leq (i_{\mathrm{old}}+1)n_{i_{\mathrm{old}}}+(i_{\mathrm{new}}+1)n_{i_{\mathrm{new}}}.$$
		Since the $U^{b_0}(i)$ is independent of everything else and does not change the message distribution, $\twc(1)$ is at most $4$ times the upper bound in \eqref{eq:localdecoder_expectedbits_proof_1}.
		Using the calculations in the proof of Lemma~\ref{lemma:localdecoder_expectedbits}, the expected number of bits to be read and written is at most $ 8b_0 $.
	\end{proof}
	
	\subsection{Proof of Theorem~\ref{theorem:main_iid} assuming that $ p_X $ is known} \label{sec:proofs}

	\subsubsection{Rate of the code} Recall that $ k_i $ is the length of a subcodeword  in the $ i $th level. The achievable rate is given by
	\[
	R = \frac{1}{n}\sum_{i=0}^{\ell_{\max}} k_i\frac{n}{n_i} = \sum_{i=0}^{\ell_{\max}}\frac{k_i}{n_i}.
	\]
	We have $ k_0\leq (H(p_X)+\varepsilon_0)b_0 $. From Definition~\ref{remark:coding_at_level_i}, we have 
	\begin{align*}
	k_i &=  b_i +\varepsilon_i n_i \log(|\cX|+1) \\
	&= n_i\left(\frac{1}{n_{i-1}}+\varepsilon_i\log(|\cX|+1)\right)\\
	&\leq n_i\left(\frac{\varepsilon_0}{2^{i(i-1)}}+\frac{\varepsilon_0}{2^i}\log(|\cX|+1)\right)\\
	&\leq n_i (1+\log(|\cX|+1))\frac{\varepsilon_0}{2^i}.
	\end{align*} 
	Therefore,
	\begin{align*}
	R&\leq H(p_X)+\varepsilon_0+(1+\log(|\cX|+1))\sum_{i=1}^{d_{\max}}\frac{\varepsilon_0}{2^i}\\
	& \leq H(p_X)+\varepsilon_0 (2+\log(|\cX|+1)). 
	\end{align*}
	Hence, the rate is $ H(p_X)+\Theta(\varepsilon_0) $.
	
	We show in Corollary~\ref{corollary:prob_error} (See Appendix~\ref{sec:preliminary_lemmas}) that the probability of error is upper bounded by $2^{-2^{O(\sqrt{\log n})}}$.
	
	\subsubsection{Average local decodability and update efficiency}
	In Lemmas~\ref{lemma:localdecoder_expectedbits} and~\ref{lemma:updating_expectedbits}, we have established that $\ravg(1)$ and $\tavg(1)$ are both $\Theta(\frac{1}{\varepsilon_0^2}\log\frac{1}{\varepsilon_0})$.
	
	Any sequence of $ s $ consecutive message symbols is spread over at most $ \lceil m/b_0\rceil +1 $ level-0 blocks.
	For any $ s \leq b_0 $, it is clear that $ \ravg(s)\leq 2\ravg(1) $.
	For $ s>b_0 $,
	\[
	\ravg(s) \leq \left(\lceil s/b_0\rceil +1\right) \alpha_{ld}b_0 = \alpha_1 s,
	\]
	for some absolute constant $ \alpha_1 $ independent of $ \varepsilon_0 $ and $ n $.
	Likewise,
	\[
	\tavg(s) = \alpha_2 s.
	\]
	for some $ \alpha_2 $ independent of $ n,\varepsilon_0 $.
	\subsubsection{Computational complexity}
	Since $ b_0 $ is a constant independent of $ n $, the total complexity for encoding/decoding all the codewords at level zero is $ \Theta(n) $. From Definition~\ref{remark:coding_at_level_i}, the computational complexity of decoding a block at level $ i $ is linear in $ n_i $, and there are $ n/n_i $ blocks at level $ i $. Since the total number of levels $ \ell_{\max} $ is $ O(\log n) $, the overall computational complexity is $ O(n\log n) $. A similar argument can be made to show that the expected computational complexity for local decoding/updating of a fragment of length $s$ is $\Theta(s)$.
	
	This completes the proof. \qed

	\subsection{Variable-length source code with zero error}
	Note that Theorem~\ref{theorem:main_iid} guarantees the existence of a fixed-length source code with a vanishing probability of error. However, in most applications, we want zero error source codes. 
	The scheme of Appendix~\ref{sec:conversion_fixlength_varlength} allows us to modify our code to give a locally decodable and update efficient variable-length compressor. 

	After the modification in Appendix~\ref{sec:conversion_fixlength_varlength}, $ \ravg(1) $ can increase by no more than $ 1 $. If the probability of error $ P_e $ is $ o(1/n) $, then the expected update efficiency also remains $ \Theta\left(\frac{1}{\varepsilon^2}\log\frac{1}{\varepsilon}\right) $. If the original fixed-length code has rate $ H(p_X)+\varepsilon $ and probability of error $ P_e $, then the new code has rate $ (1-P_e)(H(p_X)+\varepsilon)+P_e $, which  asymptotically approaches $ H(p_X)+\varepsilon $ if $ P_e =o(1)$.
	
	\subsection{Universal compression using Lempel-Ziv as a subcode}\label{sec:scheme_lz}
	
	We show that the performance by the coding scheme described above can be achieved even if the source $p_X$ is unknown to the encoder and local decoder/updater.

	Let $ \cC_i $ denote the $ (n_i,k_i/n_i) $ fixed-length compression scheme at level $ i $ in Section~\ref{sec:scheme_constant}.
	In Section~\ref{sec:scheme_constant}, we chose $\cC_0$  to be the typical set compressor. In this section, we will replace this with a fixed-length compressor based on LZ78~\cite{ziv1978compression}.

	We first redefine what it means for a sequence to be typical.
	\begin{definition}
	For any $ \delta>0 $ and $ b\in\bZ_+ $, we say that $x^b\in\cX^b $ is $ \delta $-LZ typical with respect to $ p_X $ if the length of the LZ78 codeword corresponding to $ x^b $, denoted $ \ell_{LZ}(x^b) $, is less than $ b(H(p_X)+\delta) $.     
	\end{definition}
	
	The above notion of typicality leads to a natural computationally-efficent fixed-length compression scheme.
	
	\begin{definition}[Fixed-length compression scheme derived from LZ78]
	Let $ \cT_{\delta,LZ}^b $ denote the set of all sequences that are $ \delta $-LZ typical with respect to $ p_X $. Associated with this is a natural $ (b, H(p_X)+\delta) $ fixed-length compression scheme which we denote $ \cC_{LZ}(b,H(p_X),\delta) $: 
	For any $ x^b\in\cX^b $, the corresponding codeword in $ \cC_{LZ}(b,H(p_X),\delta) $ is given by
	\[
	y^{b(H(p_X)+\delta)} = \begin{cases}
	[1, \enc_{\mathrm{LZ}}(x^b),\; 0^{b(H(p_X)+\delta)-\ell_{LZ}(x^b)}] &\text{if }\ell_{LZ}(x^b) < b(H(p_X)+\delta)\\
	0^{b(H(p_X)+\delta)} &\text{otherwise,}
	\end{cases}
	\]
	where $ \enc_{\mathrm{LZ}} $ denotes the LZ78 encoder.   
	   \label{defn:fixedlength_lz78}
	\end{definition}
	
	We can now describe the modifications required in the scheme of Section~\ref{sec:scheme_constant} in order to achieve universal compression.
	\subsubsection*{The universal compressor with locality}
	The global encoder uses the empirical estimate of $p_X$ to choose $b_0$ and $k_0$, which are encoded in the first $\Theta(1)$ bits (the preamble) of the compressed sequence\footnote{One way to store $b_0$ (resp.\ $k_0$) is by $1^{b_0}0^{k_b-b_0}$ (resp.\ $1^{k_0}0^{k_b-k_0}$) for a large enough predetermined value of $k_b=o(n)$.}. The parameter $\varepsilon_0$ can be fixed beforehand, or otherwise stored in the preamble.
	The rest of the codeword is generated as in Section~\ref{sec:scheme_constant} but with $\cC_0$ being $\cC_{LZ}$.
	
	
	

	The following theorem summarizes the main result of this section, and completes the proof of Theorem~\ref{theorem:main_iid}. The proof uses some technical lemmas that are formally proved in Appendix~\ref{sec:preliminary_lemmas_lz}.
	\begin{theorem}
		Fix a small $ \varepsilon>0 $.
		The coding scheme in Section~\ref{sec:scheme_constant} with $ \cC_0 $ chosen to be $ \cC_{LZ} $ achieves rate
		\[
		R = H(p_X)+\varepsilon,
		\]
		probability of error 
		\[
		\Pr[\dec(\enc(X^n))\neq X^n] = 2^{-2^{\Omega(\sqrt{\log n})}},
		\]
		and average local decodability and update efficiency  
		\[
		\ravg(s)\leq\begin{cases}
		\alpha_1 \frac{1}{\varepsilon^2}\log \frac{1}{\varepsilon} &\text{if } s = O\left(\frac{1}{\varepsilon^2}\log \frac{1}{\varepsilon}\right)\\
		\alpha_1's &\text{if } s = \Omega\left(\frac{1}{\varepsilon^2}\log \frac{1}{\varepsilon}\right)
		\end{cases},
		\]
		\[
		\tavg(s)\leq\begin{cases}
		\alpha_2 \frac{1}{\varepsilon^2}\log \frac{1}{\varepsilon} &\text{if } s = O\left(\frac{1}{\varepsilon^2}\log \frac{1}{\varepsilon}\right)\\
		\alpha_2's &\text{if } s = \Omega\left(\frac{1}{\varepsilon^2}\log \frac{1}{\varepsilon}\right)
		\end{cases},
		\]
		where $ \alpha_1,\alpha_1',\alpha_2,\alpha_2' $ are constants independent of $ n,\varepsilon $ but dependent on $ p_X $.
		
		The overall computational complexity of encoding and decoding $ X^n $ is $ O(n\log n) $.
		\label{theorem:main_lz}
	\end{theorem}
	\begin{proof}
		We set $ k_0 = b_0(H(p_X)+ \xi(\varepsilon_0,b_0)) $, where
		\[
		\xi(\varepsilon_0,b_0) \coloneq \left( 2+\max_{a\in\cX}\log\frac{1}{p_X(a)} \right)\varepsilon_0 +\frac{c\log\log b_0}{\log b_0}. 
		\]
		In the above, $ c $ denotes the constant that appears in Lemma~\ref{lemma:empiricalLZ}.
		Clearly, $ k_0 = b_0(H(p_X)-\Theta(\varepsilon_0)) $. At level $ 0 $, we 
		use $ \cC_0 = \cC_{LZ}(b_0,H(p_X), \xi(\varepsilon_0,b_0)) $.
		The rest of the compression scheme is exactly as in Section~\ref{main}.
		From our choice of parameters and Lemma~\ref{lemma:empiricalLZ}, it is easy to see that $ \ell_{LZ}(x^{b_0}(j))\leq k_0-1 $ as long as $ x^{b_0}(j)\in \cT_{\varepsilon_0}^{b_0} $. Therefore, the calculations in the proof of Theorem~\ref{theorem:main_iid} can be invoked to complete the proof.
		
		The rate is $ H(p_X)+\Theta(\varepsilon_0) $, while $ \ravg(s) $ and $ \tavg(s) $ are (up to constants depending only on $ p_X $) the same as in Theorem~\ref{theorem:main_iid}. 
	\end{proof}

	\section{Proof of Theorem~\ref{thm:Ologlogn}}\label{sec:o_loglogn_scheme}

We now describe our algorithm which achieves  \emph{worst-case} local decodability and \emph{average} update efficiency of $ O(\log\log n) $. The basic idea is the following: We partition the message symbols into blocks of $ O(\log\log n) $ symbols each, and compress each block using a simple variable-length compression scheme. To locate the codeword corresponding to each block, we separately store a data structure that takes $ o(n) $ space. This data structure allows us to efficiently query certain functions of the message.

For ease of exposition, we assume that $p_X$ is known. Universality can be achieved by replacing the typical set compressor in our scheme with a universal compressor such as LZ78 (as we did in Section~\ref{sec:scheme_lz}). 

\begin{definition}[Rank]
	Let $ z^m $ denote an $ m $-length binary sequence. For any $ i\in[m] $, the \emph{rank}, $ \rank_i(z^m) $ denotes the number of $ 1 $'s in (or the Hamming weight of) $ z_1^{i} $.
	
	\label{defn:rankselect}
\end{definition}

Our construction for efficient local decoding and updates is based on the existence of compressed data structures that allow query-efficient computation of rank. Let $h(\cdot)$ denote the binary entropy function.
\begin{lemma}[\cite{grossi2013dynamic}]
	Let $ m $ be a sufficiently large integer, and fix $ 0<\alpha<1/2 $.
	Then, there exists a mapping $ 
	\fsucc^{(\alpha)}:\{0,1\}^m\to \{0,1\}^{m(h(\alpha)+o(1))} $ such that for every $ x^n\in\{0,1\}^m  $ with Hamming weight at most $ \alpha m $, 
	\begin{itemize}
		\item $ x^m $ can be recovered uniquely from $ \fsucc^{(\alpha)}(x^m) $
		\item For every $ 1\leq i\leq m $, the rank $ \rank_i(x^m) $ 
		can be computed by probing at most $ O(\log m) $ bits of $ \fsucc^{(\alpha)}(x^m) $ in the worst case.
	\end{itemize}
	\label{lemma:rankselect}
\end{lemma}

\begin{figure}
	\begin{center}
		\includegraphics[width=12cm]{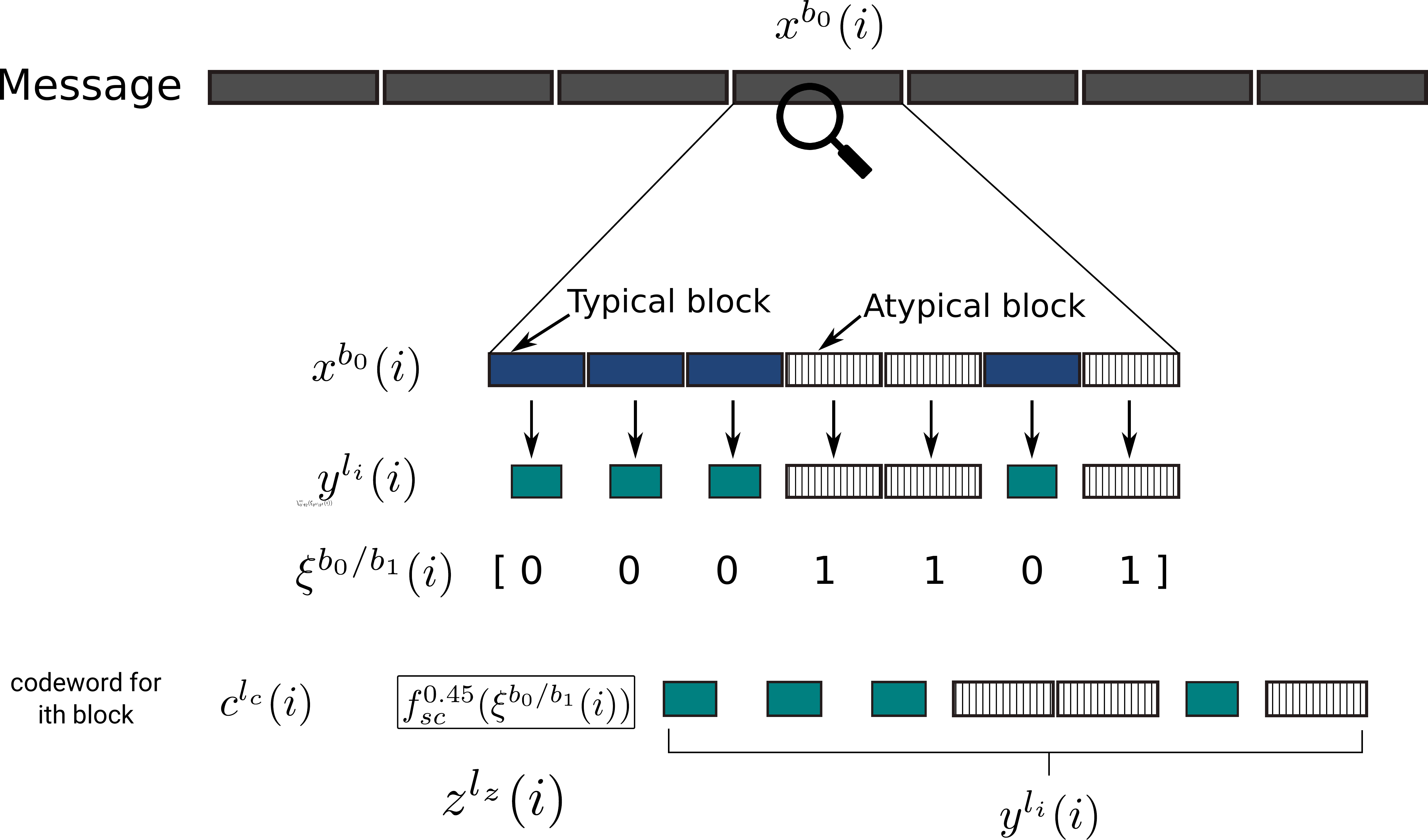}
		\caption{Compression of each block as described in Section~\ref{sec:o_loglogn_encoding}. Typical subblocks are compressed to $ \approx b_1H(p_X) $ bits, while atypical subblocks are stored without compression. The address of $ y^{l_i}(i) $ on disk can be easily computed using rank and select operations on $ \xi^{b_0/b_1}(i) $.}
		\label{fig:Ologlogn_scheme}
	\end{center}
\end{figure}

\subsubsection{Encoding}\label{sec:o_loglogn_encoding}
We partition the source sequence $ x^n $ into blocks of $ b_0=O(\log n) $ symbols each: $ x^n=(x^{b_0}(1),\ldots,x^{b_0}(n/b_0)) $. We further subdivide each block into subblocks of $ b_1 $ symbols each, {\it{i.e.}}, $ x^{b_0}(i) $ is partitioned into $ (x^{b_1}(i,1),\ldots,x^{b_1}(i,b_0/b_1)) $. The symbols $ x^{b_0}(i) $'s are encoded independently of each other using a fixed length code which has a vanishingly small probability of error. The codeword for each block consists of two parts:
\begin{itemize}
	\item Corresponding to every $ x^{b_1}(i,j) $, we generate $ y^{l_{ij}}(i,j) $, which is given by
	\[
	y^{l_{ij}}(i,j) = \begin{cases}
	\mathrm{index}(x^{b_1}(i,j);\cT_{\varepsilon_0}^{b_1}) & \text{if } x^{b_1}(i,j)\in \cT_{\varepsilon_0}^{b_1}\\
	x^{b_1}(i,j) &\text{otherwise.}
	\end{cases}
	\]
	Observe that the above is not a fixed-length code. The length of the $ (i,j) $th codeword $ l_{ij} $ is equal to $ \log |\cT_{\varepsilon_0}^{b_1}| $ if $ x^{b_1}(i,j) $ is typical and $ b_1 $ otherwise.
	Additionally, let 
	\[
	\xi(i,j) = \begin{cases}
	0 & \text{if } x^{b_1}(i,j)\in \cT_{\varepsilon_0}^{b_1}\\
	1 &\text{otherwise.}
	\end{cases}
	\]
	be an indicator of whether the $ (i,j) $th block $ x^{b_1}(i,j) $ is atypical.
	Let $ \xi^{b_0/b_1}(i)=(\xi(i,1),\ldots \xi(i,b_0/b_1)) $
	and define
	\[
	z^{\ell_{z}}(i)\defeq \begin{cases}
	\fsucc^{(\varepsilon_1)}(\xi^{b_0/b_1}(i)) &\text{if }\xi^{b_0/b_1} \text{ has Hamming weight at most }\varepsilon_0b_0/b_1 \\
	0^{\ell} &\text{otherwise,}
	\end{cases}
	\]
	where $ f_{sc} $ is the compressed data structure in Lemma~\ref{lemma:rankselect}.
	Let $ \ell_y\defeq (1-2\varepsilon_0)(H(p_X)+\varepsilon)b_0+2\varepsilon_0 b_0\log|\cX| $ and $ l_i'=\ell_y-\sum_{j}l_{ij} $
	\[
	y^{\ell_y}(i) \defeq \begin{cases}
	(y^{l_{i1}}(i,1),\ldots,y^{l_{ib_0/b_1}}(i,b_0/b_1),0^{l_i'}) &\text{if }\sum_{j}l_{ij}\leq \ell_y\\
	0^{\ell_y} &\text{otherwise}.
	\end{cases}
	\]
	The second case would correspond to an error.
	\item The codeword $ c^{\ell_c}(i) $ corresponding to $ x^{b_0}(i) $ is a sequence of length $\ell_c= \ell_y+\ell_z $, and is equal to the concatenation of $ z^{\ell_z}(i) $ and $ y^{\ell_y}(i) $.
\end{itemize}

\begin{example}[Figure~\ref{fig:Ologlogn_scheme}]
Consider the encoding of each $b_0$-length block as illustrated in Figure~\ref{fig:Ologlogn_scheme}. In this example, $b_0/b_1=7$. Subblocks $4,5,7$ are atypical. Therefore, $y^{l_{ij}}(i,j)=x^{b_1}(i,j)$ and $l_{ij}=n_1$  for $j=4,5,7$. The remaining subblocks are compressed using the typical set compressor. The indicator vector $\xi^{6}(i)=[0001101]$, and is compressed to get $z^{\ell_z}(i)$ using the scheme in Lemma~\ref{lemma:rankselect}. The overall codeword for block $i$ is the concatenation of $z^{\ell_z}(i)$  and $y^{l_{ij}}(i,j)$, $1\leq j\leq 7$.
\end{example}
\subsubsection{Local decoding of a subblock}
Our scheme allows us to locally decode an entire $b_1$-length subblock and local recovery of a single symbol is performed by locally decoding the subblock containing it.

Suppose that we want to locally decode $ x^{b_1}(i,j) $. Our local decoder works as follows:
\begin{itemize}
	\item Compute $ n_{atyp} $, the number of atypical subblocks in the first $ j $ subblocks of the $ i $th block. This is equal to $ \rank_j(\xi^{b_0/b_1}(i)) $ and can be obtained by probing $ O(\log(b_0/b_1)) $ bits of $ z^{\ell_{z}}(i) $. 
	\item Compute $ \xi(i,j) $ from $ z^{\ell_z}(i) $. This could be recovered by first decoding $ \rank_{j+1}(\xi^{b_0/b_1}(i)) $ and subtracting $ \rank_j(\xi^{b_0/b_1}(i)) $ from this. This tells us whether the block we we want to decode is atypical.
	\item Given the above information, it is easy to decode the $ (i,j) $th block. Let $ k_1 = n_{atyp} b_1+(j-1-n_{atyp})b_1(H(p_X+\varepsilon_0)) $.
	\[
	\hat{y}^{\ell_{ij}}(i,j) = \begin{cases}
	y^{k_1+b_1(H(p_X+\varepsilon_0)}_{k_1} &\text{if }\xi(i,j)=0\\
	y^{k_1+b_1}_{k_1} &\text{otherwise.}
	\end{cases}
	\]
	The estimate of the message block $ x^{b_1}(i,j) $ is obtained by decompressing $ \hat{y}^{\ell_{ij}}(i,j) $.
\end{itemize}
Let us revisit the previous example.
\begin{example}[Figure~\ref{fig:Ologlogn_scheme}]
 Figure~\ref{fig:Ologlogn_scheme}. Suppose that we are interested in recovering $x^{n_1}(i,5)$. 

The local decoder first finds $\rank_4(z^{\ell_z}(i))=1$ and $\rank_5(z^{\ell_z}(i))=2$ using the probing scheme in Lemma~\ref{lemma:rankselect}. This reveals that $x^{n_1}(i,5)$ is atypical, and one out of four subblocks prior to $x^{n_1}(i,5)$ is atypical. The starting location of $x^{n_1}(i,5)$ in $y^{\ell_i}(i)$ is $m\defeq 3n_1(H(p_X)+\varepsilon)+n_1+1$. The desired block is recoverable from $y_{m}^{m+n_1-1}(i)$.
\end{example}

\subsubsection{Update algorithm}
We consider update of $ x^{b_1}(i,j) $ with a new symbol denoted $ \widetilde{x}^{b_1} $. Let 
$$ \tilde{y}^{\ell}=\begin{cases}
(\mathrm{index}(\widetilde{x}^{b_1});\cT_{\varepsilon_0}^{b_1}) & \text{if } \widetilde{x}^{b_1}\in \cT_{\varepsilon_0}^{b_1}\\
\widetilde{x}^{b_1} &\text{otherwise.}
\end{cases} $$ 
The update algorithm works as follows:
\begin{itemize}
	\item Compute $ n_{atyp} $ and $ x^{b_1}(i,j) $ by running the local decoding algorithm above.
	\item If both $ \widetilde{x}^{b_1} $ and $ x^{b_1}(i,j) $ are typical (or both atypical), then updating the codeword is trivial as it only requires replacing $ y^{\ell_{ij}}(i,j) $ with $ \tilde{y}^{\ell} $. In this case, only $ O(\log b_0) $ bits need to be read and written in order to update the codeword.
	\item If only one of $ \widetilde{x}^{b_1} $ and $ x^{b_1}(i,j) $ is typical, then the entire code block $ c^{\ell_c}(i) $ is rewritten with the encoding of 
	$$ \widetilde{x}^{b_0}\defeq (x^{b_1}(i,1),\ldots,\widetilde{x}^{b_1},\ldots,x^{b_1}(i,b_0/b_1)). $$ 
	In this case, a total of $ O(b_0) $ bits need to be read and modified to effect the update.
\end{itemize}

\subsection{Proof of Theorem~\ref{thm:Ologlogn}}
	We choose $ b_0 = c_0\log n $ and $ b_1=c_1\log\log n $, where $ c_0 $ and $ c_1 $ are constants that need to be chosen appropriately.
	The probability that a subblock is atypical is $p_0= 2^{-\Theta(\varepsilon_0^2b_1)} $. We choose $ c_1 $ so that this probability is at most $ 1/\log^2 n $. Recall that a $ b_0 $-block is in error if more than $ 2\varepsilon_0 $ fraction of the subblocks are atypical. Using Chernoff bound, this is at most $p_1= 2^{-\Omega(b_0/b_1)} $. We can choose $ c_0 $ so as to ensure that $ p_1 $ is at most $ n^{-2} $. The probability that the overall codeword is in error is at most $ np_1 = o(1) $.
	
	We therefore have a fixed-length compression scheme with a vanishingly small probability of error. The worst-case local decodability is $\rwc(1)= \Theta(b_1) $. Updating a subblock might lead to a typical block becoming atypical (or vice versa). Therefore, the average update efficiency is 
	$$ \tavg(1) =  (1-p_1)\Theta(b_1) + p_1\Theta(b_0) = O(\log\log n).$$
	This gives the first part of the theorem.
	
	Any $ s $-length substring is contained in at most $ \lceil s/b_1\rceil +1 $ subblocks of size $b_1$. We can therefore locally decode/update any $ m $-length substring by separately running the local decoding/update algorithm for each of the $ \lceil s/b_1\rceil +1 $ subblocks. Therefore,
	\begin{align*}
	\rwc(s) &\leq \left( \left\lceil \frac{s}{b_1} \right\rceil +1 \right) \rwc(1)\\
	&\leq \begin{cases}
	2\rwc(1) &\text{if }s\leq b_1\\
	s(H(p_X)+\varepsilon)+2\rwc(1) &\text{otherwise.}
	\end{cases}
	\end{align*}
	The calculation of $ \twc(s) $ proceeds identically.
	This completes the proof of the second part of Theorem~\ref{thm:Ologlogn}.
\qed

	\section{Concluding remarks}\label{conclusion}
	In this paper, we gave an explicit, computationally efficient entropy-achieving scheme that achieves constant \emph{average} local decodability and update efficiency. Our scheme also allows efficient local decoding and update of contiguous substrings. For  $s=\Omega(1/\varepsilon^2)$, both $\ravg(s)$  and $\tavg(s)$ grow as $\Theta(s)$, where the implied constant is independent of $n$ and $\varepsilon$.

	It still remains an open problem as to whether $ (\rwc(1),\twc(1))=(\Theta(1),\Theta(1)) $ is achievable. We described a scheme with $ (\rwc(1),\tavg(1))=(O(\log\log n),O(\log\log n)) $. Even showing that $(\rwc(1),\twc(1))=(O(\log\log n),O(\log\log n))$ is achievable would be an interesting step in this direction.
	
    The careful reader might have noticed that the probability of \emph{local decoding} is nonzero, but less than or equal to the probability of global decoding, \emph{i.e.,} $\Pr[\hat{X}_i\neq X_i] \leq \Pr[\hat{X}^n\neq X^n]$. This is because the local decoder outputs the correct value of $X_i$ if $X^n$ can be recovered from $C^{nR}$. While~\cite{mazumdar2015local,tatwawadi18isit_universalRA}  achieve $\rwc(1)=\Theta(1)$, the probability of local decoding is nonzero but vanishing in $n$. Indeed, if we have a compressor that achieves zero error probability of local decoding of any single symbol, then this implies that the probability of error of global decoding is also zero (since we can run the local decoder to recover each of the $n$ symbols).
    
    It is worth pointing out that the probability of error of local decoding can influence $\rwc(s)$ significantly.
    While our scheme achieves $\ravg(1)=\tavg(1)=\Theta(1)$, we can only guarantee $\rwc(1)=\twc(1)=O(n)$.
    However, if we can tolerate a higher probability of error of local decoding (without compromising on the probability of error of global decoding), then we can achieve a smaller $\rwc(1)$, as it suffices to have the local decoder only probe the first few levels.  Specifically, if we desire $\Pr[\hat{X}_i\neq X_i]\leq \rho$ for all $1\leq i\leq n$ and some $\rho>0$, then using Lemma~\ref{lemma:prob_leftover}, we can guarantee\footnote{We do not explicitly mention the dependency on $\epsilon_0$ here.} $\rwc(1)=2^{O(\log^2(1/\rho))}$. In particular, if $\rho=\Theta(1)$, then 
	we can achieve $(\rwc(1),\tavg(1))=(\Theta(1),\Theta(1))$.
	
	Although we did not optimize the hidden constants in Theorem~\ref{theorem:main_iid}, the dependence of $\ravg(s),\tavg(s)$ on $ \varepsilon $ cannot be improved by using tighter bounds. This is because we used a lossless compression scheme at level $ 0 $, and we require $ b_0=\Omega(\frac{1}{\varepsilon^2}\log\frac{1}{\varepsilon}) $ to guarantee concentration. Mazumdar {\it{et al.}}~\cite{mazumdar2015local} used a slightly different approach, and gave a two-level construction with a lossy source code at the zeroth level. This allowed them to achieve $ \rwc(1)=\Theta(\frac{1}{\varepsilon}\log\frac{1}{\varepsilon}) $. Finding the right dependence of $ (\rwc(1),\twc(1)) $ or $ (\ravg(1),\tavg(1)) $ on $ \varepsilon $ is an interesting open question. 

\appendices

	\section{Preliminary lemmas for the proof of Theorem~\ref{theorem:main_iid}}\label{sec:preliminary_lemmas}

\begin{lemma}
	Let $ X^b$ be a $ b $-length i.i.d.\ sequence where the components are drawn according to $ p_X $. For any positive $ \alpha, $ and $ 0<\varepsilon<1/2 $, if $$ b\geq 3(
	\alpha+\log|\cX|)\left(\max_{a\in\cX}\frac{1}{p_X(a)}\right)\left(\frac{1}{\varepsilon^2}\log\frac{1}{\varepsilon}\right), $$ then
	\[
	\Pr[X^b\notin \cT^b_\varepsilon] \leq \varepsilon^\alpha.
	\]
	Moreover,
	\[
	|\cT_{\varepsilon}^b| \leq 2^{b(H(p_X)+\varepsilon)}.
	\]
	\label{lemma:typicality_stuff}
\end{lemma}
\begin{proof}
	The first part can be easily derived using Chernoff and union bounds. The second part is a standard property of typical sets. See, {\it{e.g.}}, the book by El Gamal and Kim~\cite{elgamal2011network} for a proof.
\end{proof}

\begin{lemma}
	Let $ \delta_{i-1\to i} $ denote the probability that the message block from level $ i-1 $, say $ x^{n_{i-1}}(j,i-1) $, is not the all-$\diamond $ block. If $\varepsilon_0<1/2$, $ \delta_{i-1\to i}\leq \varepsilon_i^4 $, and
	$$ b_0\geq 3(
	8+\log|\cX|)\left(\max_{a\in\cX}\frac{1}{p_X(a)}\right)\left(\frac{1}{\varepsilon_0^2}\log\frac{1}{\varepsilon_0}\right), $$
	Then,
	\begin{equation}
	\delta_{i\to i+1} \leq \varepsilon_i^{\beta 2^i},
	\label{eq:prob_leftover_1}
	\end{equation}
	where 
	\begin{equation}\label{eq:defn_beta}
	\beta = 9(
	8+\log|\cX|)\left(\max_{a\in\cX}\frac{1}{p_X(a)}\right)\left(\frac{1}{\varepsilon_0}\log\frac{1}{\varepsilon_0}\right)
	\end{equation}
	 This implies that
	\begin{equation}
	\delta_{i\to i+1} \leq \varepsilon_{i+1}^{4}.
	\label{eq:prob_leftover_2}
	\end{equation}
	\label{lemma:prob_leftover}
\end{lemma}
\begin{proof}
	Recall that a message block from level $ i $ is not a $ \diamond  $-block only if there are more than $ \varepsilon_ib_i $ non-$ \diamond  $-blocks from level $ i-1 $. Therefore,
	\begin{align}
	\delta_{i\to i+1}&\leq \nchoosek{b_i}{\varepsilon_ib_i}\delta_{i-1\to i}^{\varepsilon_i b_i}&\notag\\
	& \leq \left(\frac{\delta_{i-1\to i}}{\varepsilon_i}\right)^{\varepsilon_ib_i}&\notag\\
	&\leq \varepsilon_i^{3\varepsilon_ib_i}& \label{eq:prob_leftover_prf1}
	\end{align}
	However,
	\begin{align*}
	\varepsilon_ib_i &=\varepsilon_0 2^{-i}b_02^{2i}= \varepsilon_0b_0 2^i.
	\end{align*}
	Using the lower bound for $ b_0 $ in the above equation and substituting in (\ref{eq:prob_leftover_prf1}) gives us (\ref{eq:prob_leftover_1}). Inequality (\ref{eq:prob_leftover_2}) follows from (\ref{eq:prob_leftover_1}) by observing that $ \varepsilon_0<1/2 $.
\end{proof}

The probability of error therefore decays quasiexponentially in $ n $ as described by the following corollary.
\begin{corollary}
	Suppose we use the parameters as defined in Lemma~\ref{lemma:prob_leftover}, and choose $ b_i=2^{2i}b_0 $ and $ \varepsilon_i=\varepsilon_0/2^i $. Then, the probability that the encoder makes an error, i.e., that the message is not compressed within $ \ell_{\max} $ levels, is $ 2^{-\Omega(\ell_{\max}2^{\ell_{\max}})} $. If the number of levels is $ \Theta(\sqrt{\log n}) $, then this is $ 2^{-2^{\Omega(\sqrt{\log n})}} $.
	\label{corollary:prob_error}
\end{corollary}

The following lemma will be used to compute the average local decodability and update efficiency.
\begin{lemma}
	Let $ \delta^{(1)}_{i\to i+1} $ be the conditional probability that the message block from level $ i $, say $ x^{n_i}(1,i)$ is not the all-$\diamond $ block given that a fixed block from level $ i-1 $, say $ x^{n_{i-1}}(1,i-1) $, is not a $ \diamond  $-block. If $\varepsilon_0<1/2$, $ \delta_{i-1\to i}^{(1)}\leq \varepsilon_i^4 $, and
	$$ b_0\geq 3(
	8+\log|\cX|)\left(\max_{a\in\cX}\frac{1}{p_X(a)}\right)\left(\frac{1}{\varepsilon_0^2}\log\frac{1}{\varepsilon_0}\right). $$
	Then,
	\begin{equation}
	\delta_{i\to i+1}^{(1)} \leq \varepsilon_i^{\beta 2^{i-1}},
	\label{eq:prob_leftover_condn_1}
	\end{equation}
	where $ \beta = 9(
	8+\log|\cX|)\left(\max_{a\in\cX}\frac{1}{p_X(a)}\right)\left(\frac{1}{\varepsilon_0}\log\frac{1}{\varepsilon_0}\right) $. This implies that
	\begin{equation}
	\delta_{i\to i+1}^{(1)} \leq \varepsilon_{i+1}^{4}.
	\label{eq:prob_leftover_condn_2}
	\end{equation}
	\label{lemma:prob_leftover_condn}
\end{lemma}
\begin{proof}
	Clearly,
	\[
	\delta^{(1)}_{i\to i+1} \leq \nchoosek{b_i-1}{\varepsilon_ib_i-1}\left(\delta_{i-1\to i}^{(1)}\right)^{\varepsilon_i b_i-1}.
	\]
	The remainder of the proof is almost identical to that of Lemma~\ref{lemma:prob_leftover}, and we skip the details.
\end{proof}

The following result will be useful when bounding the average local decodability in Lemma~\ref{lemma:localdecoder_expectedbits}.
\begin{lemma}
    For all $i\geq 1$ and $b_0\geq 3$, we have
    \[
    (i+1)b_0^{i+1}2^{i(i+1)} \varepsilon_i^{\beta 2^{i-1}} \leq  \varepsilon_0^i .
    \]
    \label{lemma:locdec_intermediate_bound}
\end{lemma}
\begin{proof}
    Let $\chi(i) \defeq (i+1)b_0^{i+1}2^{i(i+1)} \varepsilon_i^{\beta 2^{i-1}}$ for $i\geq 1$. 
    
    Note that $\beta$, defined in \eqref{eq:defn_beta}, is equal to $3b_0$. Therefore,
    \[
    \chi(1) = 8b_0^2 \left( \frac{\varepsilon_0}{2} \right)^{3b_0} < \varepsilon_0,
    \]
    where the last step holds for all $b_0\geq 3$.
    For any $i\geq 2$,
    \begin{align*}
    \frac{\chi(i)}{\chi(i-1)} &= \left(\frac{i+1}{i}\right) b_02^{2i}\left(\frac{\varepsilon_0}{2^i}\right)^{3b_0(2^{i-1}-2^{i-2})} \\
    &\leq \left(\frac{i+1}{i}\right) b_02^{2i}\left(\frac{\varepsilon_0}{2^i}\right)^{6b_0}\\
    &\leq 2b_0 \left(\frac{\varepsilon_0}{2^i}\right)^{6b_0}\\
    &\leq \varepsilon_0.
    \end{align*}
    Therefore, $\xi(i)\leq \varepsilon_0^i$.
\end{proof}

\section{Preliminary lemmas for the proof of Theorem~\ref{theorem:main_lz}}\label{sec:preliminary_lemmas_lz}

In order to compute bounds on the rate and expected local decodability and update efficiency, we must find the probability that the length of an LZ78 codeword exceeds a certain amount. To help us with that, we have the following lemma:
\begin{lemma}[\cite{cover2012elements}]
	Let $ \cX $ be a finite alphabet and $ b $ be a positive integer. For any $ x^b\in\cX^b $, let $ \ell_{LZ}(x^b) $ denote the length of the LZ78 codeword for $ x^b $. For every $ k\in\bZ_+ $,  we have
	\[
	\ell_{LZ}(x^b) \leq b H_k(x^b) + \frac{ckb\log\log b}{\log b},
	\]
	where $ H_k(x^b) $ denotes the $ k $th order empirical entropy of the sequence $ x^b $, and $ c $ is an absolute constant.
	\label{lemma:empiricalLZ}
\end{lemma}
The above lemma says that the length of the LZ78 codeword is close to the empirical entropy of the string. The following lemma lets us conclude that if a sequence is typical, then the empirical entropy is close to the true entropy.
\begin{lemma}
	Fix any two probability mass functions $ p,q $ on $ \cX $, and $ 0<\varepsilon<1/2 $. If $| p(a)-q(a)|\leq \varepsilon p(a) $ for all $ a\in\cX $, then
	\[
	|H(p)-H(q)|\leq \left(2+\max_{a\in\cX}\log\frac{1}{p(a)}\right)\varepsilon.
	\]
	\label{lemma:entropy_lipschitz}
\end{lemma}
\begin{proof}
	Consider
	\begin{align*}
	\Delta_a &\coloneq p(a)\log p(a) - q(a)\log q(a) \\
	&= p(a)\log p(a) - q(a)\log p(a)+q(a)
	\log p(a)-q(a)\log q(a)\\
	&= (p(a)-q(a))\log p(a) -q(a)\log \frac{q(a)}{p(a)}
	\end{align*}
	However,
	\begin{align*}
	|H(p)-H(q)|&\leq \sum_a |\Delta_a|\\ 
	&\leq \sum_a\left(|p(a)-q(a)|\log\frac{1}{p(a)} +q(a)\left|\log \frac{q(a)}{p(a)}\right|\right) \\
	&\leq \varepsilon \max_a\log\frac{1}{p(a)} + \log\frac{1}{1-\varepsilon}.
	\end{align*}
	For $ \varepsilon<1/2 $, we have $ \log\frac{1}{1-\varepsilon} \leq 2\varepsilon $. Using this in the above completes the proof.
\end{proof}

\section{Fixed v/s Variable-length compression}
We briefly show how to achieve zero-error data compression and still achieve the performance stated in Theorems~\ref{theorem:main_iid} and \ref{thm:Ologlogn}. This is obtained by using a variable length code
instead of a fixed-length code.
\begin{definition}[Variable-length compression]
	An $ (n,R) $ variable-length compression scheme is a pair of maps $ (\enc,\dec) $ consisting of 
	\begin{itemize}
		\item an encoder $ \enc: \cX^n\to \{0,1\}^{*} $, and
		\item a decoder $ \dec: \{0,1\}^{*}\to \cX^n $ satisfying
		\[
		\dec(\enc(X^n)) = X^n, \quad \forall X^n\in\cX^n
		\]
	\end{itemize} 
	For any $ Y^l\in \{0,1\}^* $, let $ \ell(Y^l) $ denote the length of the sequence $ Y^l $. 
	The quantity $ R $ is the rate of the code, and is defined to be
	\[
	R\defeq \frac{1}{n} \bE [\ell(\enc(X^n))]
	\]
	where the averaging is over the randomness in the source.

\end{definition}
It is generally desired for a variable-length source code be prefix free: For every distinct pair of inputs $ X^n,Y^n\in\cX^n $, the codeword $ \enc(X^n) $ must not be a prefix of $ \enc(Y^n) $.

\subsubsection{Converting a fixed-length compressor to a prefix-free variable-length compressor}\label{sec:conversion_fixlength_varlength}
Given any $ (n,R) $ fixed-length compression scheme $ (\enc_{\mathrm{fix}},\dec_{\mathrm{fix}}) $ with a probability of error $ P_e = o(1) $, it is easy to construct a prefix-free $ (n,R+o(1)) $ variable-length compressor $ (\enc_{\mathrm{var}},\dec_{\mathrm{var}}) $. The following is one-such construction:
\[
\enc_{\mathrm{var}}(X^n)\defeq \begin{cases}
(0,\enc_{\mathrm{fix}}(X^n)) &\text{if } \dec_{\mathrm{fix}}(\enc_{\mathrm{fix}}(X^n)) = X^n\\
(1,X^n) & \text{otherwise}.
\end{cases}
\] 
Clearly, the compressor is prefix free.
The rate of $ (\enc_{\mathrm{var}},\dec_{\mathrm{var}}) $ is equal to 
\begin{align*}
R_{\mathrm{var}} &= 1/n+ R\times(1-P_e) + \log|\cX|\times P_e\\
&= R + o(1).
\end{align*}
For all $s\geq 1$, the local decodability of the new variable-length scheme $\ravg(s),\rwc(s)$ is at most $1$ more than that of the original fixed-length scheme, and $\tavg(s),\twc(1)$ can increase by at most $2$.

Due to the above transformation, we have devoted most of our attention to constructing fixed-length compression schemes.


	\bibliographystyle{IEEEtran}
	\bibliography{locality_references}
\end{document}